\documentclass[preprint,review,12pt]{elsarticle}

\usepackage[width=15.2cm,height=21.5cm,centering]{geometry}
\usepackage{color}
\usepackage{graphics}
\usepackage{amsmath}
\usepackage{amssymb}
\usepackage{amsthm}
\usepackage{times}
\usepackage[captionskip=-1mm]{subfig}
\usepackage{cases}
\usepackage{enumitem}
\usepackage[colorlinks,citecolor=blue,linkcolor= blue]{hyperref}

\usepackage{lineno}

\newcommand{\eps}{\varepsilon}
\newcommand{\sig}{\sigma}
\newcommand{\real}{\operatorname{Re}}
\newcommand{\imag}{\operatorname{Im}}
\newcommand{\p}{\partial}
\newcommand{\td}{\,\text{d}}
\newcommand{\fu}{F_{v}}
\newcommand{\fs}{F_{\sigma}}
\newcommand{\tta}{\tilde{\theta}}
\newcommand{\tmu}{\tilde{\mu}}
\newcommand{\Ula}{\Upsilon_{\ell,\alpha}}
\newcommand{\gla}{\gamma_{\ell,\alpha}}
\newcommand{\bU}{\boldsymbol{U}}
\newcommand{\bbU}{\overline{\boldsymbol{U}}}
\newcommand{\bA}{\boldsymbol{A}}

\newcommand{\bS}{\boldsymbol{S}}
\newcommand{\bbS}{\overline{\boldsymbol{S}}}
\newcommand{\bF}{\boldsymbol{F}}
\newcommand{\bbF}{\overline{\boldsymbol{F}}}
\newcommand{\bI}{\boldsymbol{I}}
\newcommand{\trsp}{^{\mbox{\tiny$\mathsf{T}$}}}
\newcommand{\inv}{^{\mbox{\tiny$-1$}}}
\newcommand{\indL}{_{\mbox{\tiny$L$}}}
\newcommand{\expL}{^{\mbox{\tiny$L$}}}

\newcommand{\expLpt}{^{\mbox{\tiny$L$}+2}}
\newcommand{\indLpt}{_{\mbox{\tiny$L$}+2}}
\newcommand{\indLpo}{_{\mbox{\tiny$L$}+1}}

\newcommand{\indLa}{_{\mbox{\tiny$L$},\alpha}}

\newcommand{\spc}{\operatorname{sp}}
\newcommand{\sgn}{\operatorname{sgn}}

\newcommand{\tr}{\operatorname{tr}}
\newcommand{\tchi}{\tilde{\chi}}

\newcommand{\bHp}{\boldsymbol{\mathcal{H}}_p}
\newcommand{\bHd}{\boldsymbol{\mathcal{H}}_d}

\newtheorem{proposition}{Proposition}

\newtheorem{remark}{Remark}
\newtheorem*{remark*}{Remark}

\journal{Wave Motion}

\begin{document}

\begin{frontmatter}

\title{Wave propagation in a fractional viscoelastic Andrade medium: diffusive approximation and numerical modeling}

\author[ECM]{A. Ben Jazia\corref{cor1}}
\ead{abderrahim.ben-jazia@centrale-marseille.fr}
\author[LMA]{B. Lombard\corref{cor2}}
\ead{lombard@lma.cnrs-mrs.fr}
\author[LMA]{C. Bellis\corref{cor3}}
\ead{bellis@lma.cnrs-mrs.fr}
\cortext[cor2]{Corresponding author}
\address[ECM]{Ecole Centrale Marseille, 13451 Marseille, France}
\address[LMA]{LMA, CNRS, UPR 7051, Aix-Marseille Univ., Centrale Marseille, F-13402 Marseille Cedex 20, France}


\begin{abstract}
This study focuses on the numerical modeling of wave propagation in fractionally-dissipative media. These viscoelastic models are such that the attenuation is frequency dependent and follows a power law with non-integer exponent. As a prototypical example, the Andrade model is chosen for its simplicity and its satisfactory fits of experimental flow laws in rocks and metals. The corresponding constitutive equation features a fractional derivative in time, a non-local term that can be expressed as a convolution product which direct implementation bears substantial memory cost. To circumvent this limitation, a diffusive representation approach is deployed, replacing the convolution product by an integral of a function satisfying a local time-domain ordinary differential equation. An associated quadrature formula yields a local-in-time system of partial differential equations, which is then proven to be well-posed. The properties of the resulting model are also compared to those of the original Andrade model. The quadrature scheme associated with the diffusive approximation, and constructed either from a classical polynomial approach or from a constrained optimization method, is investigated to finally highlight the benefits of using the latter approach. Wave propagation simulations in homogeneous domains are performed within a split formulation framework that yields an optimal stability condition and which features a joint fourth-order time-marching scheme coupled with an exact integration step. A set of numerical experiments is presented to assess the efficiency of the diffusive approximation method for such wave propagation problems.
\end{abstract}


\end{frontmatter}


\linenumbers

\section{Introduction}\label{SecIntro}

There is a long history of studies discussing or providing experimental evidences of frequency-dependent viscoelastic attenuations, as observed in e.g. metals \cite{Andrade:1910iq}, acoustic media \cite{Szabo:94,Szabo:95} and in the Earth \cite{Flanagan,Lekic}. Such a behavior is classically modeled using a fractional derivative operator \cite{Torvik,Carcione}, a mathematical tool generalizing to real parameters the standard derivatives of integer orders \cite{Podlubny}. While fractional calculus is now a mature theory in the field of viscoelasticity \cite{Mainardi}, some issues remain commonly encountered. They mostly revolve around the two questions of: (1) incorporating fractional dissipation into viscoelastic models that both fit experimental data and have a theoretical validity regarding, e.g., causality properties \cite{Holm,Waters:2005} or the Kramers-Kronig relations \cite{Waters}; and (2) implementing numerically these fractional models to perform wave propagation simulations. The latter problem is commonly tackled using standard approaches \cite{Emmerich} for modeling constant-law of attenuation over a frequency-band of interest, i.e. with the fractional viscoelastic model being approximated by multiple relaxation mechanisms \cite{Nasholm}.

Bearing in mind the issue (1) discussed above, it is chosen to anchor the present study to a specific, yet prototypical, physically-based viscoelastic model, namely the Andrade model. Initially introduced in \cite{Andrade:1910iq} to fit experimental flow laws in metals, it has been further investigated in \cite{Andrade:1962uh}. It is now used as a reference in a number of studies \cite{Gribb:1998tf,Jackson:2010iv,Sundberg:2010p4848,Bellis13a} for the description of observed frequency-dependent damping behaviors in the field of geophysics and experimental rock mechanics. Moreover, the Andrade model creep function, as written, can notably be decomposed into a fractional power-law added to a standard Maxwell viscoelastic model. Therefore, while being physically motivated and rooted in experiments, this model gives leeway to cover the spectrum from a conventional rheological mechanism to a more complex fractional model, and this with only a few parameters.\\

This study focuses on the issue (2), namely the numerical modeling of wave propagation within a fractionally-dissipative Andrade medium. The objective is to develop an efficient approximation strategy of the fractional term featured in this viscoelastic model in view of the investigation and simulation of its transient dynamical behavior. A model-based approximation approach is explored in the sense that the original constitutive equation is not intended to be superseded by another viscoelastic model, \emph{per se}, which would be designed to fit only a given observable, such as the quality factor, or overall behavior.\\

 The article aim and contribution are twofold:
\begin{enumerate}[leftmargin=0cm,itemindent=14mm,labelwidth=8mm,labelsep=0mm,align=left,label={(\roman*)},noitemsep,topsep=0pt] 
\item Deploy an approximation of the fractional derivative featured in the constitutive equation considered. A direct discretization of this term, that is associated with a non-local time-domain convolution product \cite{Podlubny} requires the storage of the entire variables history. While being potentially improved by Gr\"unwald-Letnikov finite-difference approximation schemes, see e.g. \cite{Galucio}, such an approach remains costly numerically and is therefore circumvented. Alternatively, a so-called diffusive representation is preferred \cite{Matignon:proc}, as it allows to recast the equations considered into a local-in-time system while introducing only a limited number of additional memory variables in its discretized form \cite{Yuan02}. Following later improvements of the method in \cite{Diethelm08,Birk10,Haddar,Deu}, an efficient quadrature scheme is investigated in order to obtain a satisfactory fit of the reference model compliance.
\item Implement the resulting approximated model into a wave propagation scheme. While the available literature on the numerical simulation of transient wave propagation within fractionally-damped media is relatively scarce, see e.g. \cite{Wismer,Caputo11}, the aim is here to demonstrate the practicality and efficiency of the proposed approach. For the sake of simplicity, the viscoelastic medium considered is assumed to be unidimensional and homogeneous. After discretization of the dynamical system at hand, a Strang splitting approach \cite{LeVeque92} is adopted, both to reach an optimal stability condition and to enable the use of an efficient high-order time-marching scheme coupled with an exact integration step. Moreover, deriving a semi-analytical solution for the configuration considered, as a baseline, a set of numerical results is presented to assess the quality of the numerical scheme developed. The overall features and performances of the diffusive representation are finally discussed to compare the original Andrade model with its diffusive approximated counterpart.\\
\end{enumerate}

This article is organized as follows. The reference Andrade model is presented and discussed in Section \ref{SecModel}. Considering the featured fractional derivative, a corresponding diffusive approximated version of the former is subsequently formulated and referred to as the Andrade--DA model. The evolution problem is investigated in Section \ref{SecPDE}, with the derivation and analysis of the first-order hyperbolic system associated with the Andrade--DA model. Section \ref{SecNumQuad} is concerned with the definition and computation of an efficient quadrature scheme for the diffusive approximation, while the implementation of the fully discretized system is described in Section \ref{SecNumScheme}. Corresponding numerical results are presented and discussed in Section \ref{SecExp}.


\section{Fractional viscoelastic model}\label{SecModel}

\subsection{Preliminaries}\label{SecModelPreliminaries}

The causal constitutive law describing the behavior of a 1D linear viscoelastic medium can be expressed in terms of the time-domain convolution
\begin{equation}
\eps(t)=\int_0^t \chi(t-\tau)\frac{\p \sig}{\p \tau}(\tau)\td\tau,
\label{Creep}
\end{equation}
with creep function $\chi$, stress field $\sig$ and strain field $\eps={\p u}/{\p x}$ associated with unidimensional displacement $u$.

Next, for parameters satisfying $0<\beta<1$, the so-called Caputo-type fractional derivative \cite{Carcione,Mainardi,Podlubny} of a causal function $g(t)$ is defined as
\begin{equation}\label{def:frac:diff}
\frac{\td^{\beta}g}{\td t^\beta}(t)=\frac{1}{\Gamma(1-\beta)}\int_{0}^{t}(t-\tau)^{-\beta}\frac{\td g}{\td \tau}(\tau)\td\tau,
\end{equation}
where $\Gamma$ is the Gamma function. Defining the direct and inverse Fourier transforms in time of a function $g(t)$ as
\[
\hat{g}(\omega)=\int_{-\infty}^{+\infty}{g(t)e^{-i\omega t}\td t}, \qquad
g(t)=\frac{1}{2\pi}\int_{-\infty}^{+\infty}\hat{g}(\omega)e^{i\omega t}\td\omega,
\]
then the frequency-domain counterpart of equation \eqref{def:frac:diff} reads
\begin{equation}
\widehat{\left[\!\frac{\td^{\beta}g}{\td t^\beta}\right]}\!(\omega)=(i\omega)^\beta\,\hat g(\omega),
\label{FourierFrac}
\end{equation}
so that definition \eqref{def:frac:diff} is a straightforward generalization of the derivative of integer order.


\subsection{Andrade model}\label{SecModelAndrade}

The Andrade model \cite{Andrade:1910iq} is characterized by the creep function given by
\begin{equation}
\label{CreepAndrade}
\chi(t)=\bigg[J_u+\frac{t}{\eta}+A\,t^\alpha\bigg] H(t), \hspace{1cm} 0<\alpha<1,
\end{equation}
with Heaviside step function $H(t)$, unrelaxed compliance $J_u$, viscosity $\eta$ and two positive physical parameters $A$ and $\alpha$. Usual fits with experimental data correspond to $\frac13\leq\alpha\leq\frac12$ \cite{Andrade:1962uh,Gribb:1998tf}. The composite law \eqref{CreepAndrade} can be additively decomposed into a standard Maxwell rheological mechanism with creep function $t\mapsto J_u+t/\eta$ and a relaxation time $\tau_{\mbox{\tiny Mx}}=\eta\,J_u$, together with a power law dependence in time $t\mapsto A\,t^{\alpha}$ which constitutes its main feature. Examples behaviors of the creep function \eqref{CreepAndrade} are illustrated in Figure \ref{fig:creep}.

The frequency-domain compliance ${ N}=i\omega\hat{\chi}$, and such that $\hat\eps={N}\,\hat\sig$ according to the Fourier transform of equation \eqref{Creep}, can be deduced from \eqref{CreepAndrade} as
\begin{equation}
{ N}(\omega)=J_u-\frac{ i}{ \eta\,\omega}+A\,\Gamma(1+\alpha)\,(i\,\omega)^{-\alpha}.
\label{Compliance}
\end{equation}
Straightforward manipulations on \eqref{FourierFrac}, \eqref{CreepAndrade} and \eqref{Compliance} lead to the following constitutive equation in differential form for the Andrade model
\begin{equation}
\frac{\p\eps}{\p t}=J_u\frac{\p\sig}{\p t}+\frac{1}{\eta}\sig+A\,\Gamma(1+\alpha)\frac{\p^{1-\alpha}}{\p t^{1-\alpha}}\sig.
\label{stress:strain:And2}
\end{equation}


\subsection{Dispersion relations}\label{SecModelWave}


\begin{figure}[htbp]	
\centering
\subfloat[Creep function $\chi(t)$]{\includegraphics[width=0.5\textwidth]{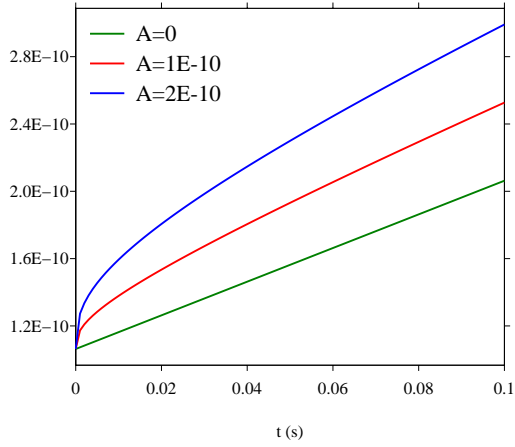}\label{fig:creep}}
\subfloat[Quality factor $Q(f)$]{\includegraphics[width=0.5\textwidth]{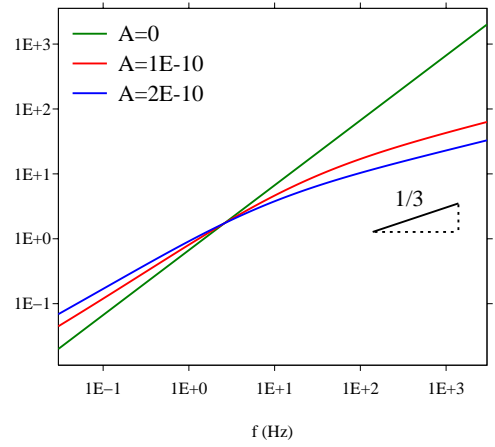}\label{fig:Q}}\\    
\subfloat[Phase velocity $c(f)$]{\includegraphics[width=0.5\textwidth]{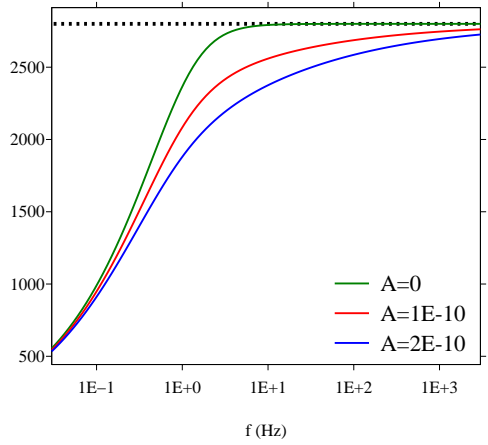}\label{fig:Cel}}
\subfloat[Attenuation $\zeta(f)$]{\includegraphics[width=0.5\textwidth]{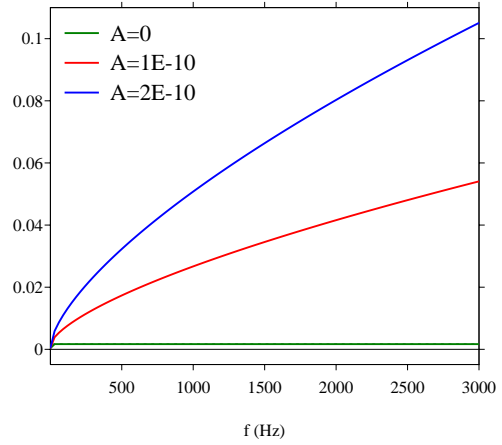}\label{fig:Att}}
\caption{Behaviors of various viscoelastic models derived from \eqref{CreepAndrade}: Maxwell model ($A=0$) and Andrade model ($\alpha=1/3$, with $A=10^{-10}$ {Pa}$^{-1}.${s}$^{-\alpha}$ and $A=2.10^{-10}$ {Pa}$^{-1}.${s}$^{-\alpha}$). The other physical parameters are: $\rho=1200$ kg/m$^3$, $c_{\infty}=2800$ m/s and $\eta=10^9$ Pa.s. The horizontal dotted line in panel (c) denotes the high-frequency limit $c_{\infty}$.}
\label{FigDispersion}
\end{figure}

The complex wave number $k(\omega)$ satisfies
\begin{equation}
k(\omega)=\sqrt{\rho}\,\omega\, {[N(\omega)]}^{1/2}:=\frac{\omega}{c(\omega)}-i\zeta(\omega),
\label{KAndrade}
\end{equation}
where the phase velocity $c$ and the attenuation $\zeta$ are given by
\begin{equation}
c(\omega)=\sqrt{\frac{2}{\rho(|{N}|+\real[{N}])}},\qquad\zeta(\omega)=\omega\sqrt{\frac{\rho(|{N}|-\real[{N}])}{2}}.
\label{CelAndrade}
\end{equation}
Owing to equations \eqref{Compliance} and \eqref{CelAndrade}, the following limits hold:
\begin{equation}
\begin{aligned}\label{LimitsBfHf}
&\lim_{\omega \to 0} {c(\omega)} =0 , & \quad  &  \lim_{\omega \to +\infty} {c(\omega)} =\frac{1}{\sqrt{\rho J_u}}:= c_{\infty},\\
&\lim_{\omega \to 0} {\zeta(\omega)} =0 ,  & \quad  &  \lim_{\omega \to +\infty} {\zeta(\omega)} =+\infty.
\end{aligned}
\end{equation}

Moreover, when $A> 0$, the creep function \eqref{CreepAndrade} is an increasing and concave function. As a consequence, owing to the theoretical developments in \cite{Hanyga13a} and \cite{Hanyga13b}, the attenuation $\zeta(\omega)$ for the Andrade model turns out to be sublinear in the high-frequency range, i.e.
\begin{equation}
\zeta(\omega)\underset{\omega\to+\infty}{=}o(\omega).
\label{AttConcave}
\end{equation}
This key property confirms the relevance of the choice of the Andrade model as a prototypical example of fractional viscoelastic media.

The quality factor $Q$ is defined as the ratio
\begin{equation}
Q(\omega)=-\frac{\real[k^2]}{\imag[k^2]}=-\frac{\real[{N}]}{\imag[{N}]}.
\label{Q:def}
\end{equation}
According to \eqref{Compliance} and in a high-frequency regime, the frequency-dependent behavior follows
\begin{equation}
Q(\omega)\underset{\omega\to+\infty}{\sim} Q_\infty\,\omega^\alpha \qquad \text{ with } Q_\infty={ J_u}\bigg[ A\,\Gamma(1+\alpha)\,\sin\!\left( \frac{ \alpha\pi}{ 2}\right)\bigg]^{-1}.
\label{Qinfty}
\end{equation}
Sample behaviors of the Andrade model for $\alpha=1/3$ and a varying parameter $A$ are sketched in Figure \ref{FigDispersion}. Notably, the case $A=0$ corresponds to the standard Maxwell model. When $A\neq0$, one observes in Fig. \ref{fig:Q} the slope $1/3$ of the quality factor in log-log scale at high frequencies, as expected from \eqref{Qinfty}. The attenuation $\zeta$ is represented as a function of the frequency $f$ and displayed in linear scale in Fig. \ref{fig:Att} to emphasize the sublinear high-frequency behavior \eqref{AttConcave}.


\subsection{Diffusive approximation: Andrade--DA model}\label{SecModelDA}

When implementing \eqref{stress:strain:And2}, the difficulty revolves around the computation of the convolution product in \eqref{def:frac:diff} associated with the fractional derivative of order $1-\alpha$, which is numerically memory-consuming. The alternative approach adopted in this study is based on a diffusive representation, and its approximation, of fractional derivatives. Following \cite{Diethelm08}, then for $0<\alpha<1$ equation \eqref{def:frac:diff} can be recast as
\begin{equation}\label{Diff:Rep}
\frac{\p^{1-\alpha}}{\p t^{1-\alpha}}\sigma=\int_{0}^{+\infty}\!\!{\phi(x,t,\theta) \td\theta},
\end{equation}
where the function $\phi$ is defined owing to a change of variables as
\begin{equation}
\phi(x,t,\theta)=\frac{2\sin(\pi\alpha)}{\pi} \theta^{1-2\alpha}\int_{0}^t\frac{\p \sigma}{\p \tau}(x,\tau)\,e^{-(t-\tau)\,\theta^2}\td \tau.
\label{def:mem:var}
\end{equation}
As $\phi$ is expressed in terms of an integral operator with decaying exponential kernel, it is referred to as a \emph{diffusive} variable. From equation \eqref{def:mem:var}, it can be shown to satisfy the following first-order differential equation for $\theta>0$:
\begin{equation}\label{ode:DR}
\left\{
\begin{aligned}
& \frac{\p \phi}{\p t}=-\theta^2\,\phi+\frac{2\sin(\pi\alpha)}{\pi} \theta^{1-2\alpha}\frac{\p \sigma}{\p t},\\
& \phi(x,0,\theta)=0.\\
\end{aligned}
\right.
\end{equation}
The diffusive representation (\ref{Diff:Rep}--\ref{def:mem:var}) amounts to supersede the non-local term in \eqref{stress:strain:And2} by an integral of the function $\phi(x,t,\cdot)$ obeying the local first-order ordinary differential equation \eqref{ode:DR}. The integral featured in \eqref{Diff:Rep} is in turn well-suited to be approximated using a quadrature scheme, so that
\begin{equation}\label{Diff:Approx}
\frac{\p^{1-\alpha}}{\p t^{1-\alpha}}\sigma\simeq\sum_{\ell=1}^{L} {\mu_\ell\, \phi(x,t,\theta_\ell)}\equiv\sum_{\ell=1}^{L} {\mu_\ell\, \phi_\ell(x,t)},
\end{equation}
given a number $L$ of quadrature nodes $\theta_\ell$ with associated weights $\mu_\ell$, whose computations will be returned to in Section \ref{SecNumQuad}.\\

The frequency-domain versions of equations \eqref{stress:strain:And2}, \eqref{ode:DR} and \eqref{Diff:Approx} lead to the approximated complex compliance ${ {\tilde N}}$, such that $\hat\eps={ {\tilde N}}\,\hat\sig$ and characterizing the model hereafter referred to as the Andrade--DA model, as
\begin{equation}
{ {\tilde N}}(\omega)=J_u-\frac{ i}{ \eta\,\omega}+A \Gamma(1+\alpha)\,\frac{2 \sin(\pi\alpha)}{\pi} \sum_{\ell=1}^{L}\mu_\ell\,\frac{\theta_{\ell}^{1-2\alpha}}{\theta_{\ell}^2+i\omega}.
\label{ComplianceDA}
\end{equation}
The diffusive approximated counterparts of \eqref{KAndrade} and \eqref{CelAndrade} can be immediately deduced from \eqref{ComplianceDA}. In particular, the low-frequency and high-frequency limits of the phase velocity ${\tilde c}$ are equal to those in \eqref{LimitsBfHf}. Moreover, using tables of standard Fourier transforms, the corresponding time-domain creep function $\tchi$, defined by $\tilde{N}=i\omega\hat{\tchi}$, is obtained as
\begin{equation}
\tchi(t)=\bigg[J_u +\frac{t}{\eta}+ A \Gamma(1+\alpha) \frac{2 \sin(\pi\alpha)}{\pi} \sum_{\ell=1}^{L} {\mu_\ell\, \theta_\ell^{-1-2\alpha} \left( 1-e^{-\theta_\ell^2t} \right)}\bigg] H(t).
\label{CreepDA}
\end{equation}


\section{Evolution equations}\label{SecPDE}

With the complex compliance \eqref{ComplianceDA} of the Andrade--DA model at hand, which constitutes the approximated version of the diffusive representation of the original Andrade model \eqref{Compliance}, the present section is concerned with the description and analysis of its dynamical behavior.

\subsection{First-order system}\label{SecPDEsys}

Let define the parameters
\begin{equation}\label{def:gUla}
\gla=\frac{2\sin(\pi\alpha)}{\pi\,J_u}\theta_\ell^{1-2\alpha}, \qquad \Ula=A\Gamma(1+\alpha)\,\gla  \qquad  \text{for }\ell=1,\dots,L.
\end{equation}
Combining the conservation of momentum in terms of velocity field $v={\p u}/{\p t}$ and equations \eqref{stress:strain:And2}, \eqref{ode:DR} and \eqref{Diff:Approx} yields 
\begin{subnumcases}{\label{EDP}}
\frac{\p v}{\p t} -\frac{1}{\rho} \frac{\p \sigma}{\p x}=\fu,\label{EDP1}\\
\frac{\p \sigma}{\p t} -\frac{1}{J_u} \frac{\p v}{\p x}=-\frac{1}{J_u \eta}\sig-\frac{A \Gamma(1+\alpha)}{J_u}\sum_{j=1}^{L} {\mu_j \phi_j}+\fs,\label{EDP2}\\
\frac{\p \phi_\ell}{\p t}- \gla \frac{\p v}{\p x}  = -\theta_\ell^2 \phi_\ell -\frac{\gla}{\eta}\sigma - \Ula \sum_{j=1}^{L} {\mu_j \phi_j}+J_u\gla\,\fs,\label{EDP3}
\end{subnumcases}
for $\ell=1,\dots,L$ and where $\fu$ and $\fs$ are introduced to model external sources. Equations \eqref{EDP} are completed by initial conditions
\[
v(x,0)=0,\qquad \sig(x,0)=0, \qquad \phi_\ell(x,0)=0 \qquad \text{ for }\ell=1,\cdots,L.
\]
Gathering unknown and sources terms, let the vectors $\bU$ and $\bF$ be defined as
\begin{equation}
\bU=\big[v,\,\sigma,\,\phi_1,\,\cdots,\,\phi\indL\big]\trsp, \hspace{1cm} \bF=\big[\fu,\,\fs,\,J_u\gamma_{1,\alpha}\,\fs,\,\cdots,\,J_u\gamma\indLa\,\fs\big]\trsp.
\label{VecUF}
\end{equation}
Then the system \eqref{EDP} can be written in the matrix-form
\begin{equation}\label{syst:AU}
\frac{\p \bU}{\p t} +\bA \frac{\p \bU}{\p x}= \bS\bU+\bF,
\end{equation}
where $\bA$ is given by
\begin{equation}\label{def:bA}
\bA=\left[\begin{array}{ccccc}
0 & -\rho\inv & 0 & \cdots&0 \\
-{J_u}\inv & 0 &0&\cdots&0 \\
-\gamma_{1,\alpha}  & 0 & 0&\cdots&0\\
\vdots&\vdots&\vdots&\ddots&\vdots \\
-\gamma\indLa & 0 & 0&\cdots&0\\
\end{array}\right],
\end{equation}
and $\bS$ reads
\begin{equation}\label{def:bS}
\bS=\left[\begin{array}{ccccc}
0 &0 &0&\cdots&0\\[1mm]
0 & - (J_u \eta)\inv & -{A\Gamma(1+\alpha)}{J_u}\inv\mu_1& \cdots & -{A\Gamma(1+\alpha)}{J_u}\inv\mu\indL\\[1mm]
0 & -\gamma_{1,\alpha}\eta\inv & -\theta_1^2-\Upsilon_{1,\alpha}\,\mu_1 & \cdots & -\Upsilon_{1,\alpha}\, \mu\indL \\[1mm]
\vdots & \vdots & \vdots & \ddots & \vdots  \\[1mm]
0 & -\gamma\indLa\eta\inv& -\Upsilon\indLa\,\mu_1 & \cdots & -\theta\indL^2-\Upsilon\indLa\,\mu\indL
\end{array}\right].
\end{equation}
Note that this differential system remains valid in the case of a \emph{non-homogeneous} viscoelastic medium.


\subsection{Energy decay}\label{SecPDEnrj}

Studying the energy associated with the system \eqref{EDP} is required to characterize the stability of the Andrade--DA model and to provide constraints on the diffusive approximation calculation. For an infinite 1D domain, the stored kinetic and elastic energies are defined as
\[ 
{\cal E}_v(t)= \frac{1}{2} \int_{-\infty}^{+\infty}{\rho v^2 \td x} \quad \text{and} \quad {\cal E}_\sig(t)= \frac{1}{2} \int_{-\infty}^{+\infty}{J_u \sigma^2 \td x},
\]
together with a coupled term associated with the diffusive approximation
\[
{\cal E}_d(t)=\frac{1}{2} \int_{-\infty}^{+\infty}{\sum_{\ell=1}^{L}{\frac{\mu_\ell\Ula}{\theta_\ell^2} \left(\sqrt{J_u}\,\sigma-\frac{\phi_\ell}{\sqrt{J_u}\,\gla}\right)}^{\!2} \td x}.
\]
Then, in the absence of any source term, one has the following property
\begin{proposition}\label{PropNRJ}
If $\mu_\ell>0$ for all $\ell=1,\dots,L$, then the function ${\cal E}(t)={\cal E}_v(t)+{\cal E}_\sig(t)+{\cal}_d(t)$ is a positive definite quadratic form and $\normalfont\displaystyle\frac{\td {\cal E}}{\td t}<0$ for all time $t>0$.
\end{proposition}

\begin{proof}
In the absence of any source term, then multiplying the momentum equation \eqref{EDP1} by the velocity field $v$ and integrating spatially by part yields
\[
\int_{-\infty}^{+\infty}\bigg\{\rho v \frac{\p v}{\p t} + \sigma \frac{\p v}{\p x}\bigg\}\td x=0,
\]
assuming that the elastic fields vanish at infinity. Likewise, from equation \eqref{EDP2} and multiplying by $\sig$, one obtains
\begin{equation}\label{ab}
\frac12 \frac{\td}{\td t} \int_{-\infty}^{+\infty}\left\{\rho v^2 + J_u \sigma^2\right\}\td x + \int_{-\infty}^{+\infty}\bigg\{{\frac{\sigma^2}{\eta}} + A\Gamma(1+\alpha) \sum_{\ell=1}^{L}{\mu_\ell\, \phi_\ell \,\sigma}\bigg\}\td x=0. 
 \end{equation}
Now, using twice differential equation \eqref{ode:DR}, one has for $\ell=1,\dots,L$
\[
\sigma \frac{\p{\phi_\ell}}{\p t}+\theta_\ell^2 \phi_\ell \sigma -J_u\gla \sigma \frac{\p \sigma}{\p t}=0\quad \text{and}\quad \frac{\phi_\ell}{J_u\gla} \frac{\p{\phi_\ell}}{{\p t}}+\frac{\theta_\ell^2 \phi_\ell^2 }{J_u\gla}- \phi_\ell \frac{\p \sigma}{\p t}=0,
\]
which after subtraction and manipulation entails
\begin{equation}\label{final}
\phi_\ell\,\sigma=\frac{\phi_\ell^2}{J_u\gla} +\frac{\gla}{2\theta_\ell^2}  \frac{\td}{\td t}  \bigg( \sqrt{J_u}\,\sigma -\frac{\phi_\ell}{\sqrt{J_u}\,\gla} \bigg)^{\!2}.
\end{equation}
Finally, substituting \eqref{final} in \eqref{ab} leads to the relation
\begin{multline*}
\frac12 \frac{\td}{\td t}  \int_{-\infty}^{+\infty}\bigg\{\rho v^2 + J_u \sigma^2 + \sum_{\ell=1}^L\frac{\mu_\ell\Ula}{\theta_\ell^2}\bigg(\sqrt{J_u}\,\sigma -\frac{\phi_\ell}{\sqrt{J_u}\,\gla}\bigg)^{\!2}\,\bigg\}\td x\\
=-\int_{-\infty}^{+\infty}\bigg\{\frac{\sigma^2}{\eta}+\sum_{\ell=1}^L\mu_\ell\Ula\left(\frac{\phi_\ell}{\sqrt{J_u}\,\gla}\right)^{\!2}\bigg\}\td x,
\end{multline*}
which concludes the proof.
\end{proof}

In summary, positivity of the quadrature nodes and weights in \eqref{Diff:Approx} is crucial to ensure the well-posedness of the system \eqref{EDP}. This issue will be further discussed in Section \ref{SecNumQuad}.
 

\subsection{Properties of matrices}\label{PDEmat}

Some properties of the matrices $\bA$ \eqref{def:bA} and $\bS$ \eqref{def:bS} are discussed to characterize the first-order system \eqref{syst:AU} of partial differential equations.

\begin{proposition}\label{prop:vp:A}The eigenvalues of the matrix $\bA$ are
\[
\spc(\bA)=\left\{0,\pm c_{\infty}\right\},\qquad \text{with $0$ being of multiplicity $L$}.
\]
\end{proposition}
\noindent As $\bA$ is diagonalizable with real eigenvalues, then equation \eqref{syst:AU} is a hyperbolic system of partial differential equations, with solutions of finite-velocity. It is emphasized that the eigenvalue $c_\infty=1/\sqrt{\rho J_u}$ does not depend on the quadrature coefficients $\{(\mu_\ell,\theta_\ell)\}_\ell$, so that the phase velocity upper bounds for the Andrade and Andrade--DA models are equal.

\begin{proposition}
Assuming $\theta_\ell>0$ and $\mu_\ell>0$ for $\ell=1,\dots,L$ then $\spc(\bS)\ni 0$ with multiplicity $1$. Moreover the $L+1$ non-zero eigenvalues $\lambda_\ell$ of $\bS$ are real and, ordering the nodes as $0<\theta_1<\cdots<\theta\indL$, satisfy
\[
\lambda\indLpo<-\theta\indL^2<\cdots<-\theta_\ell^2<\lambda_\ell<-\theta_{\ell-1}^2<\cdots<\lambda_1<0.
\]
\label{PropVpS}
\end{proposition}

\begin{proof}
Let ${\cal P}_{\!\bS}(\lambda)$ denote the characteristic polynomial of the matrix $\bS$, i.e. ${\cal P}_{\!\bS}(\lambda)=\det(\bS-\lambda\bI\indLpt)$ with $\bI\indLpt$ the $(L+2)$-identity matrix. The line $i$ and the column $j$ of the determinant are denoted by ${\cal L}_i$ and ${\cal C}_j$, respectively. The following algebraic manipulations are performed successively: 
\begin{itemize}
\item[(i)] ${\cal L}_j\leftarrow {\cal L}_j-\gamma_\alpha\,\theta_j^{1-2\alpha}\,{\cal L}_1\qquad\text{with }j=2,\dots, L+1$
\item[(ii)] ${\cal C}_1\leftarrow {\cal C}_1 \prod\limits_{\ell=1}^L(-\theta_\ell^2-\lambda)$
\item[(iii)] ${\cal C}_1\leftarrow {\cal C}_1-\gamma_\alpha\,\theta_\ell^{1-2\alpha}\,\lambda\,{\cal C}_\ell\prod\limits_{\substack{i=1\\i\neq \ell}}^L(-\theta_i^2-\lambda)\qquad\text{for }\ell=2,\dots, L+1$.
\end{itemize}
From \eqref{def:bS} and definition \eqref{def:gUla} of parameters $\gla$ and $\Ula$, one deduces
\[
{\cal P}_{\!\bS}(\lambda)=\lambda\Bigg[\big((J_u \eta)\inv+\lambda\big)\prod_{\ell=1}^L(-\theta_\ell^2-\lambda)+\lambda\sum_{\ell=1}^L\mu_\ell\Ula\prod_{\substack{j=1\\j\neq \ell}}^L(-\theta_j^2-\lambda)\Bigg]:=\lambda\, {\cal Q}_{\bS}(\lambda).
\]
From the above equation, one has ${\cal P}_{\!\bS}(0)\neq0$ while ${\cal Q}_{\bS}(0)\neq0$, therefore $0$ is an eigenvalue of the matrix $\bS$ with multiplicity $1$. In the limit $\lambda\to0$, then asymptotically 
\begin{equation}\label{sgn1}
{\cal P}_{\!\bS}(\lambda)\underset{\lambda\to0}{\sim}(-1)\expL\,(J_u \eta)\inv\,\lambda\,\prod\limits_{\ell=1}^L\theta_\ell^2, \quad\text{so that}\quad \sgn({\cal P}_{\!\bS}(0^-))=(-1)^{L+1}.
\end{equation}
Moreover, using \eqref{def:gUla} and the assumptions considered, then at the quadrature nodes one has for all $k\in\{1,\dots,L\}$
\[
{\cal P}_{\!\bS}(-\theta_k^2)=-\frac{ 2\sin(\pi\alpha)A\,\Gamma(1+\alpha)}{\pi J_u}\,\mu_k\,\theta_k^{5-2\alpha}\prod_{\substack{j=1\\j\neq k}}^L(\theta_k^2-\theta_j^2)\,\Rightarrow\, \sgn({\cal P}_{\!\bS}(-\theta_k^2))=(-1)^{L-k+1}.
\]
Finally, the following limit holds
\begin{equation}\label{sgn3}
{\cal P}_{\!\bS}(\lambda)\underset{\lambda\to-\infty}{\sim}(-1)\expL\lambda\expLpt \quad \Rightarrow\quad \sgn({\cal P}_{\!\bS}(-\infty))=1.
\end{equation}
We introduce the following intervals
\begin{equation}\label{def:int}
\mathcal{I}_{\indLpo}=\big]-\infty,-\theta\indL^2\big], \quad \mathcal{I}_{\ell+1}=\big]-\theta_{\ell+1}^2,-\theta_\ell^2\big]\quad\text{for }\ell=1,\dots,L-1 \quad\text{and}\quad \mathcal{I}_1=\big]-\theta_1^2,0\big].
\end{equation}
Given that $\lambda\mapsto {\cal P}_{\!\bS}(\lambda)$ is continuous, then equations (\ref{sgn1}--\ref{sgn3}) show that the polynomial ${\cal P}_{\!\bS}$ changes sign in each of the intervals $\mathcal{I}_\ell$ of \eqref{def:int}. Consequently, there exist $\lambda_\ell\in \mathcal{I}_\ell$ with $\ell=1,\dots,L+1$ such that ${\cal P}_{\!\bS}(\lambda_\ell)=0$ and which coincide with the eigenvalues, with multiplicity $1$, of the matrix $\bS$ of size $L+2$. \end{proof}

Proposition \ref{PropVpS} states that, under suitable conditions on the quadrature coefficients, the matrix $\bS$ in \eqref{def:bS} has eigenvalues with negative or zero real parts. This property is crucial regarding the numerical modeling developed in the forthcoming Section \ref{SecNumScheme}. As for the energy analysis given in Proposition \ref{PropNRJ}, positivity of quadrature nodes and weights is again the fundamental hypothesis. Lastly, it is possible to use the above proposition to characterize the spectral radius of the matrix $\bS$. 

\begin{proposition}\label{prop:sp:rad:S} 
The spectral radius of the matrix $\bS$ \eqref{def:bS} is such that
\[
\max\!\bigg(\theta\indL^2\,,(J_u \eta)\inv+\sum_{\ell=1}^L\mu_\ell\Ula \bigg) \leq \varrho(\bS) \leq \theta\indL^2 + (J_u \eta)\inv+\sum_{\ell=1}^L \mu_\ell\Ula. 
\]
\end{proposition}

\begin{proof}
By definition, one has
\begin{equation}\label{tr:S}
\tr(\bS)=-\bigg[ (J_u \eta)\inv+\sum_{\ell=1}^L(\theta_\ell^2+\mu_\ell\Ula) \bigg]\equiv\sum_{\ell=1}^{L+1}\lambda_\ell.
\end{equation}
According to the proof of Property \ref{PropVpS}, the eigenvalues $\lambda_\ell$ satisfy
\[
-\sum_{\ell=1}^{L}\theta_{\ell}^2 \leq \sum_{\ell=1}^L\lambda_\ell \leq -\sum_{\ell=1}^{L-1}\theta_\ell^2.
\]
Substitution in \eqref{tr:S} and providing that $\varrho(\bS)=|\lambda_{\indLpo}|$ allows to concludes the proof. 
\end{proof}


\subsection{Semi-analytical solution}\label{SecPDEana}

Let us consider a homogeneous Andrade--DA medium described by equations \eqref{EDP1} and \eqref{EDP2}, together with equation \eqref{ode:DR}. A corresponding semi-analytical solution is sought in order to validate the ensuing numerical simulations of wave propagation. It is assumed $\fs=0$ and excitation $\fu(x,t)=F(t)\delta(x-x_s)$ at source point $x_s$ with time evolution $F$. Applying space-time Fourier transforms and their inverses leads to the stress field solution in the form of
\[
\hat{\sig}(x,\omega)= \frac{i \hat{F}(\omega)}{2\pi c_{\infty}^2 {J_u} }\int_{-\infty}^{+\infty} \! \frac{ \displaystyle k}{\displaystyle k^2 -k_0^2}\, e^{i k (x-x_s)} \td k :=   \frac{i\hat{F}(\omega)}{2\pi c_{\infty}^2 {J_u} }\int_{-\infty}^{+\infty} g(k)\td k,
\]
with
\[
 k_0=\left[ \left(\! \frac{\omega}{c_{\infty}}\! \right)^{\!\!2}\left[ 1 +  \sum_{\ell=0}^{L}\frac{\mu_\ell \Ula }{\theta_\ell^2+i\omega} \right] - \frac{i\rho\,\omega}{\eta} \right]^{1/2}.
\]
The poles $\pm k_0$ of $g$ are simple and satisfy $\imag[k_0]<0$. Using the residue theorem, one obtains in the time-domain the stress field solution
\begin{equation}\label{sol:sig}
\sig(x,t)=-\frac{\sgn(x-x_s)}{2\pi c_\infty^2 J_u} \int_{0}^{+\infty} \real\!\Big[\hat{F}(\omega) e^{i(\omega t -k_0|x-x_s|)}\Big]\td \omega.
\end{equation}
Similarly, the velocity field and the memory variables satisfy
\begin{equation}\label{sol:v:phi}\begin{aligned}
& v(x,t)=\frac{1}{2\pi}\int_{0}^{+\infty}\real\!\bigg[\frac{k_0}{\omega}\hat{F}(\omega)e^{i(\omega t -k_0|x-x_s|)}\bigg]\td \omega,  \\
& \phi_\ell(x,t)=-\frac{\sgn(x-x_s)\gla}{2\pi c_\infty^2} \int_{0}^{+\infty}\real\!\bigg[\frac{i\omega}{\theta_\ell^2+i\omega}\hat{F}(\omega) e^{i(\omega t -k_0|x-x_s|)}\bigg]\td \omega, \quad \ell=1,\dots,L.
\end{aligned}\end{equation}
In the numerical results presented Section \ref{SecExp}, the frequency-domain integrals featured in solutions \eqref{sol:sig} and \eqref{sol:v:phi} are computed using a standard quadrature rule over the frequency-band considered.


\section{Numerical methods}\label{SecNum}

\subsection{Quadrature methods}\label{SecNumQuad}

Two different approaches can be employed to determine the $2L$ coefficients $\{(\mu_\ell,\theta_\ell)\}_\ell$ of the diffusive approximation \eqref{Diff:Approx}. While the most usual one is based on orthogonal polynomials, the second approach is associated with an optimization procedure applied to the model compliance. Both lead to positive quadrature coefficients, which ensures the stability of the Andrade--DA model, as shown by propositions \ref{PropNRJ} and \ref{PropVpS}.

\paragraph{Gaussian quadrature.} Various orthogonal polynomials can be used to evaluate the improper integral \eqref{Diff:Rep} introduced by the diffusive representation of fractional derivatives. Historically, the first one has been proposed in \cite{Yuan02}, where a Gauss-Laguerre quadrature is chosen. Its slow convergence was highlighted and then corrected in \cite{Diethelm08} with a Gauss-Jacobi quadrature. This latter method has been lastly modified in \cite{Birk10}, where alternative weight functions are introduced, yielding an improved discretization of the diffusive variable owing to the use of an extended interpolation range. Following this latter modified Gauss-Jacobi approach, while omitting the time and space coordinates for the sake of brevity, the improper integral \eqref{Diff:Rep} is then recast as
\begin{equation}
\int_0^{+\infty}\phi(\theta)\td\theta=\int_{-1}^{+1}\big(1-\tta\big)^{\!\gamma}\!\big(1+\tta\big)^{\!\delta}\tilde{\phi}(\tta)\td \tta \simeq\sum_{\ell=1}^L \tmu_{\ell}\,\tilde{\phi}(\tta_{\ell}),
\label{GJ1}
\end{equation}
with the modified diffusive variable $\tilde{\phi}$ defined as
\[
\tilde{\phi}(\tta)=\frac{ 4}{ \big(1-\tta\big)^{\!\gamma-1}\big(1+\tta\big)^{\!\delta+3}}\,\phi\Bigg(\!\bigg(\frac{ 1-\tta}{ 1+\tta}\bigg)^{\!\!2}\Bigg),
\]
and where the weights and nodes $\{(\tmu_{\ell},\tta_{\ell})\}_\ell$ can be computed by standard routines \cite{NRPAS}. According to the analysis of \cite{Birk10}, Section 4, an optimal choice for the coefficients in \eqref{GJ1} is in the present case: $\gamma=3-4\alpha$ and $\delta=4\alpha-1$. Finally, the quadrature coefficients of \eqref{Diff:Approx} are identified as
\begin{equation}
\mu_{\ell}=\frac{ 4\,\tmu_{\ell}}{ \big(1-\tta_{\ell}\big)^{\!\gamma-1}\big(1+\tta_{\ell}\big)^{\!\delta+3}}, \qquad\theta_{\ell}=\bigg(\frac{ 1-\tta_{\ell}}{ 1+\tta_{\ell}}\bigg)^{\!\!2}.
\label{Birk2}
\end{equation}


\paragraph{Optimization quadrature.} Alternatively, the quadrature coefficients can be deduced from the model physical observables. Note that as the quality factor \eqref{Q:def} is defined as the ratio $Q(\omega)=-\real[{N}]/\imag[{N}]$, then optimizing an objective function based on the latter would entail an indetermination on the function $N$, i.e. of the model constitutive equation. Therefore, a direct optimization of the available Andrade model compliance $N$ is preferred.

From \eqref{Compliance} and its diffusive approximated counterpart \eqref{ComplianceDA}, the corresponding compliances $N$ and $\tilde{N}$ differ only in the terms
\[
\left\{
\begin{aligned}
& \kappa(\omega):=(i\omega)^{-\alpha} && \qquad \textit{Andrade},\\
& {\tilde \kappa}(\omega):=\frac{2 \sin(\pi\alpha)}{\pi} \sum_{\ell=1}^{L}\mu_\ell\,\frac{\theta_{\ell}^{1-2\alpha}}{\theta_{\ell}^2+i\omega}\, && \qquad \textit{Andrade--DA}.
\end{aligned}
\right.
\]
For a given number $K$ of angular frequencies $\omega_k$, one introduces the objective function
\begin{equation}
\mathcal{J}\Big({\{(\mu_\ell,\theta_\ell)\}}_\ell\,;L,K\Big)=\sum_{k=1}^K\left|\frac{{\tilde \kappa}(\omega_k)}{ \kappa(\omega_k)}-1\right|^2=\sum_{k=1}^K\left|\frac{2 \sin(\pi\alpha)}{\pi} \sum_{\ell=1}^{L}\mu_\ell\,\frac{\theta_{\ell}^{1-2\alpha}\,(i\omega_k)^{\alpha}}{\theta_{\ell}^2+i\omega_k}-1\right|^2
\label{obj:opt}
\end{equation}
to be minimized w.r.t parameters $\{(\mu_\ell,\theta_\ell)\}_\ell$ for $\ell=1,\dots,L$. 

A straightforward linear minimization of \eqref{obj:opt} may lead to some negative parameters \cite{TheseBlanc,Blanc13} so that a nonlinear optimization with the positivity constraints $\mu_{\ell}\geq 0$ and $\theta_{\ell}\geq 0$ is preferred. The additional constraint $\theta_{\ell}\leq \theta_{\text{max}}$ is also introduced to avoid the algorithm to diverge. These $3L$ constraints can be relaxed by setting $\mu_{\ell}={\mu'_{\ell}}^{2}$ and $\theta_{\ell}={\theta'_{\ell}}^{2}$ and solving the following problem with only $L$ constraints
\begin{equation}{\label{OptiNL}}
\underset{{\{(\theta'_{\ell},\mu'_{\ell})\}}_\ell}{\min}\mathcal{J}\Big({\{({\mu'_\ell}^2\!,{\theta'_\ell}^2)\}}_\ell\,;L,K\Big)\qquad
\text{with }{\theta'_{\ell}}^2\leq \theta_{\text{max}}\text{ for }\ell=1,\dots,L.
\end{equation}
As problem \eqref{OptiNL} is nonlinear and non-quadratic w.r.t. abscissae $\theta'_{\ell}$, we implement the algorithm SolvOpt \cite{Kappel00, Rekik11} based on the iterative Shor's method \cite{Shor85}. Initial values $\mu'^{\,0}_{\ell}$ and $\theta'^{\,0}_{\ell}$ used in the algorithm must be chosen with care; for this purpose we propose to use the coefficients obtained by the modified Jacobi method \eqref{Birk2} for $\ell=1,\dots,L$
\begin{equation}\label{init:val}
\mu'^{\,0}_{\ell}=\sqrt{\frac{ 4\,\tmu_{\ell}}{ \big(1-\tta_{\ell}\big)^{\!\gamma-1}\big(1+\tta_{\ell}\big)^{\!\delta+3}}},\qquad \theta'^{\,0}_{\ell}=\frac{ 1-\tta_{\ell}}{ 1+\tta_{\ell}}.
\end{equation}
Finally, the angular frequencies $\omega_k$ for $k = 1,...,K$ in \eqref{obj:opt} are chosen linearly on a logarithmic scale over a given optimization band $[\omega_{\text{min}},\omega_{\text{max}}]$, i.e.
\begin{equation}\label{equi:distrib}
\omega_k = \omega_{\text{min}}\left( \frac{\omega_{\text{max}}}{\omega_{\text{min}}}\right)^{\!\frac{k-1}{K-1}}.
\end{equation}

\begin{remark}
In the proposed optimization method, both set of quadrature coefficients $\mu_\ell$ and $\theta_\ell$ are computed by minimization of the objective function $\mathcal{J}$. In particular, the nodes $\theta_\ell$ are not imposed to be equidistributed according to \eqref{equi:distrib} as it is the case in the commonly used approach \cite{Emmerich}. This point will be returned to in Section \ref{SecExpQuad}.
\end{remark}


\subsection{Numerical scheme}\label{SecNumScheme}

A numerical scheme is proposed to compute the solution of system \eqref{syst:AU}. Introducing a uniform grid with mesh size $\Delta x$ and time step $\Delta t$, let $\bU_{\!j}^{n}$ denotes the approximation of the solution $\bU(x_j\!=\!j\Delta x,t_n\!=\!n\Delta t)$ with $j=1,\dots,N_x$ and $n=1,\dots,N_t$. Straightforward discretization of \eqref{syst:AU} typically yields to the numerical stability condition \cite{TheseBlanc}
\[
\Delta t\leq \min\left(\frac{\Delta x}{c_\infty},\,\frac{2}{\varrho(\bS)}\right).
\]
As shown by Proposition \ref{prop:sp:rad:S}, the usual CFL bound on the time step $\Delta t\leq \Delta x/c_\infty$ may be reduced as $\eta$ decreases or $A$ increases, which turns out to be detrimental to the numerical scheme. Moreover, as $\varrho(\bS)$ depends on the quadrature coefficients of the diffusive variable the stability condition would in turn not depend only on meaningful physical quantities such as the maximum phase velocity $c_\infty$.

\paragraph{Splitting.} Alternatively, we follow here the splitting approach analyzed in \cite{LeVeque92}. To implement \eqref{syst:AU} numerically, one solves successively the propagative equation 
\begin{equation}
\frac{\p \bU}{\p t} +\bA \frac{\p \bU}{\p x}= \boldsymbol{0}
\label{SplitPropa}
\end{equation}
and the diffusive equation
\begin{equation}
\frac{\p \bU}{\p t} = \bS\bU+\bF.
\label{SplitDiffu}
\end{equation}
Due to the structure of matrix $\bS$, one defines from \eqref{VecUF} the subvectors
\begin{equation}\label{subvec}
\bbU=\big[\sigma,\,\phi_1,\,\cdots,\,\phi\indL\big]\trsp, \hspace{1cm} \bbF=\big[\fs,\,J_u\gamma_{1,\alpha}\,\fs,\,\cdots,\,J_u\gamma\indLa\,\fs\big]\trsp,
\end{equation}
and from \eqref{def:bS} the submatrix
\[
 \bbS=\left[\begin{array}{cccc}
 - (J_u \eta)\inv & -{A\Gamma(1+\alpha)}{J_u}\inv\mu_1& \cdots & -{A\Gamma(1+\alpha)}{J_u}\inv\mu\indL\\[1mm]
 -\gamma_{1,\alpha}\eta\inv & -\theta_1^2-\Upsilon_{1,\alpha}\,\mu_1 & \cdots & -\Upsilon_{1,\alpha}\, \mu\indL \\[1mm] 
 \vdots & \vdots & \ddots & \vdots  \\[1mm]
 -\gamma\indLa\eta\inv& -\Upsilon\indLa\,\mu_1 & \cdots & -\theta\indL^2-\Upsilon\indLa\,\mu\indL
\end{array}\right].
\]
Having separated the two source terms, then equation \eqref{SplitDiffu} is equivalently recast in the form
\begin{subnumcases}{\label{SplitDiffuRecast}}
\frac{\p v}{\p t}=\fu,\\
[8pt]
\frac{\p \bbU}{\p t} = \bbS\,\bbU+\bbF.
\end{subnumcases}
The discrete operators associated with the discretizations of \eqref{SplitPropa} and \eqref{SplitDiffuRecast} are respectively denoted by $\bHp$ and $\bHd$. The operator $\bHd$ depends explicitly on time when the forcing terms $\fu$ or $\fs$ are non-zero, whereas $\bHp$ remains independent on $t$. The so-called Strang splitting approach of \cite{LeVeque92} is then used between time steps $t_n$ and $t_{n+1}$, for $n=0,\dots,N_t-1$, which requires to solve \eqref{SplitPropa} and \eqref{SplitDiffu} with adequate time increments as, for $j=1,\dots,N_x$
\begin{equation}
\begin{aligned}
& \bU_{\!j}^{(1)} && = && \bHd(t_n,\,{\Delta t}/{2})\,\bU_{\!j}^{n}, \\[2mm]
& \bU_{\!j}^{(2)} && = &&  \bHp(\Delta t,j)\,\bU^{(1)},  \\[2mm]
& \bU_{\!j}^{n+1} && = && \bHd(t_{n+1},{\Delta t}/{2})\,\bU_{\!j}^{(2)},
\end{aligned}
\label{AlgoSplitting}
\end{equation}
with $\quad\bU^{(1)}=\big[\bU_{\!1}^{(1)}\dots\bU_{\!N_x}^{(1)}\big]\trsp$. Since the matrices $\bA$ and $\bS$ do not commute, an error associated with the splitting scheme is introduced \cite{LeVeque92}. However, provided that $\bHp$ and $\bHd$ are at least second-order accurate and stable, then the time-marching scheme \eqref{AlgoSplitting} constitutes a second-order accurate approximation of the original equation \eqref{syst:AU}.


\subparagraph{Diffusive operator.} The physical parameters do not vary with time, thus the matrix $\bbS$ does not depend on $t$. Owing to Property \ref{PropVpS}, one has $0\notin\spc(\bbS)=\{\lambda_1,\dots,\lambda\indL\}$, and hence $\det \bbS\neq 0$. Freezing the forcing terms at $t_k$, with $k=n$ or $n+1$, yields for a generic vector $\bU_{\!j}=[v_j,\,\bbU_{\!j}]\trsp$ 
\begin{equation}
\bHd(t_k,{\Delta t}/{2})\,{\bU}_{\!j}=\left[v_j+\frac{\Delta t}{2}\fu(x_j,t_k),\,e^{\bbS\frac{\Delta t}{2}}\bbU_{\!j}-\left(\bI-e^{\bbS\frac{\Delta t}{2}}\right)\bbS^{-1}\bbF(x_j,t_k)\right]\trsp.
\label{HdForcing}
\end{equation}
If there is no excitation, i.e. $\fu=\fs=0$, then integration \eqref{HdForcing} is exact. The matrix exponential entering the definition of the operator $\bHd$ is computed using the method $\sharp$2 in \cite{Moler} based on a (6/6) Pad\'e approximation. Property \ref{PropVpS} ensures that the computation of this exponential is stable.
  

\subparagraph{Propagative operator.} To integrate \eqref{SplitPropa}, we use a fourth-order ADER (Arbitrary DERivative) scheme \cite{Schwartzkopff04}. This explicit two-step and single-grid finite-difference scheme writes
\begin{equation}\label{def:ADER}
\bU_{\!j}^{(2)}=\bU_{\!j}^{(1)}-\sum_{\ell=-2}^{\ell=2}\sum_{m=1}^{4}{\vartheta_{m,\ell} \left( \bA \frac{\Delta t}{\Delta x} \right)^{\!\!m}}\,\bU_{\!j+\ell}^{(1)}:=\bHp(\Delta t,j)\,\bU^{(1)},
\end{equation}
where the coefficients $\vartheta_{m,k}$ are provided in Table \ref{ader:prm}. It satisfies the optimal stability condition $c_\infty\,\Delta t\,/\,\Delta x \leq 1$.

\begin{table}[ht]
\centering
\begin{tabular}{c|cccc}
& $m=1$ & $m=2$ & $m=3$ & $m=4$ \\ 
\hline
$\ell=-2$ & $1/12$ &$ 1/24$ & $-1/12$ & $-1/24 $\\
$\ell=-1 $& $-2/3 $& $-2/3 $& $1/6$ &$ 1/6$ \\
$\ell=0 $&  $0$ & $5/4$ & $0$ &$ 1/4$\\
$\ell=1 $& $2/3$ & $-2/3$ & $-1/6$ & $1/6$ \\
$\ell=2$ & $-1/12$ & $1/24$ & $1/12$ & $-1/24$ \\
\end{tabular}
\caption{Coefficients $\vartheta_{m,\ell}$ in the ADER--4 scheme \eqref{def:ADER}\vspace{-2mm}}\label{ader:prm}
\end{table}


\section{Numerical results}\label{SecExp}

\subsection{Configuration}\label{SecExpConfig}

The homogeneous domain considered is $400$ m-long and it is characterized by the physical parameters provided in Table \ref{TabParam} and which are consistent with experimentally-based values, see \cite{Bellis13a} and the references therein.

\begin{table}[ht]
\begin{center}
\begin{small}
\begin{tabular}{|c|c|c|c|c|}
\hline
$\rho$ (kg/m$^3$) & $c_{\infty}$ (m/s) & $\eta$ (Pa.s) & $A$ ({Pa}$^{-1}.${s}$^{-\alpha}$) & $\alpha$ \\
\hline  
1200              & 2800               & $10^9$        & $2\,10^{-10}$                     & 1/3      \\
\hline
\end{tabular}
\end{small} 
\end{center}
\vspace{-0.5cm}
\caption{Chosen physical parameters in the Andrade model \eqref{CreepAndrade}.}
\label{TabParam}
\end{table}


\subsection{Validation of the quadrature methods}\label{SecExpQuad}

The angular frequency range of interest $[\omega_{\text{min}},\omega_{\text{max}}]$ is defined by $\omega_{\text{min}}=\omega_c/100$ and $\omega_{\text{max}}=10\,\omega_c$ for a given central angular frequency ${\omega_c}={60\,\pi}$, while we set $\theta_{\text{max}}=\sqrt{100\,\omega_c}$ and $K=2\,L$. Observables of the original Andrade model \eqref{Compliance} are then compared to those of the Andrade--DA model \eqref{ComplianceDA} on Figure \ref{FigApprox} for the two quadrature methods discussed in Section \ref{SecNumQuad}. Large deviations are observed when the Gaussian quadrature is used, in particular on the attenuation function. On the contrary, an excellent agreement between the original Andrade model and its diffusive counterpart is obtained. Only slight differences can be observed at the scale of the figures within the optimization interval.

\begin{figure}[th!]	
\centering
\subfloat[Creep function $\chi(t)$]{\includegraphics[width=0.5\textwidth]{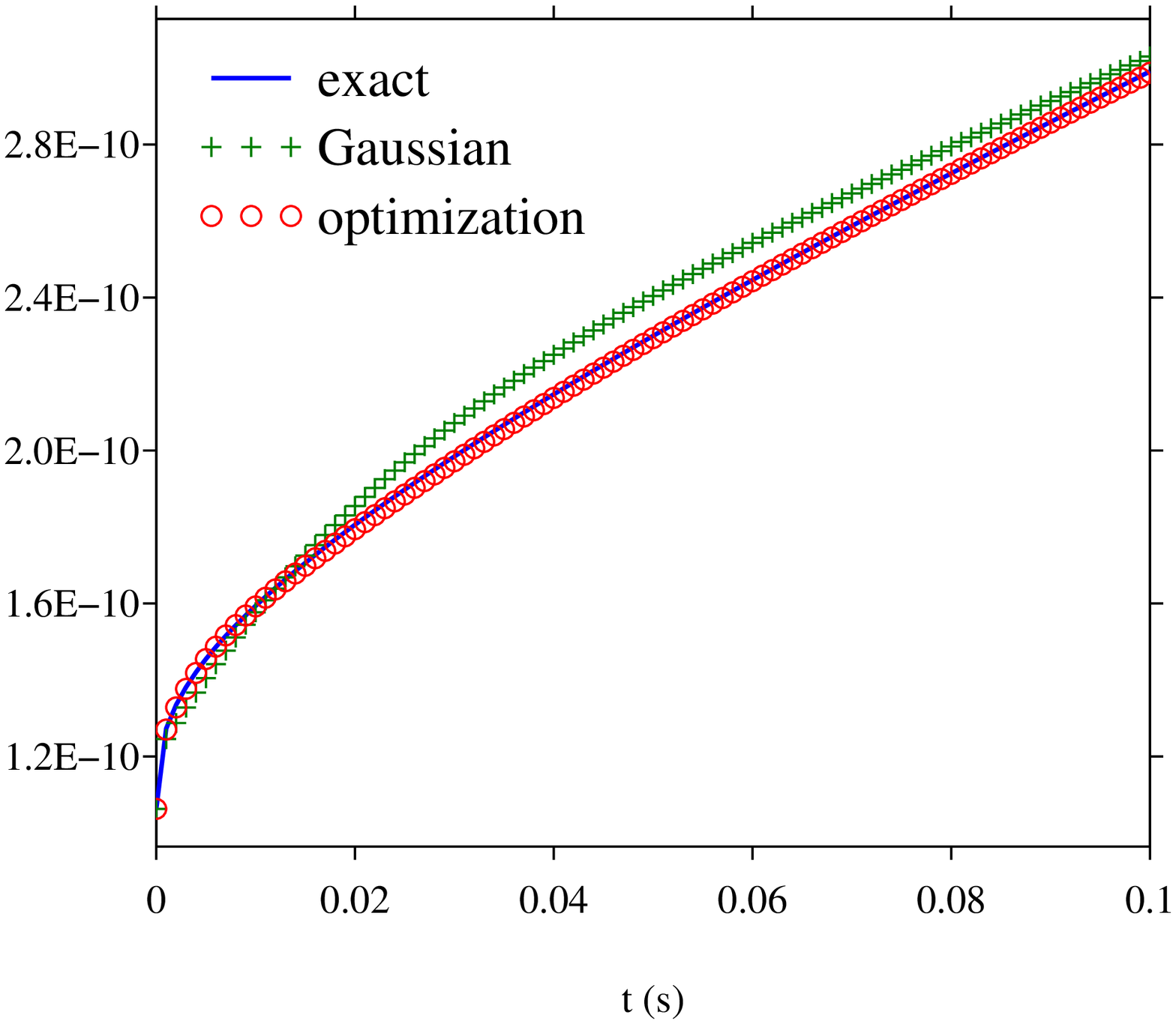}\label{FluOpti4}}
\subfloat[Quality factor $Q(f)$]{\includegraphics[width=0.5\textwidth]{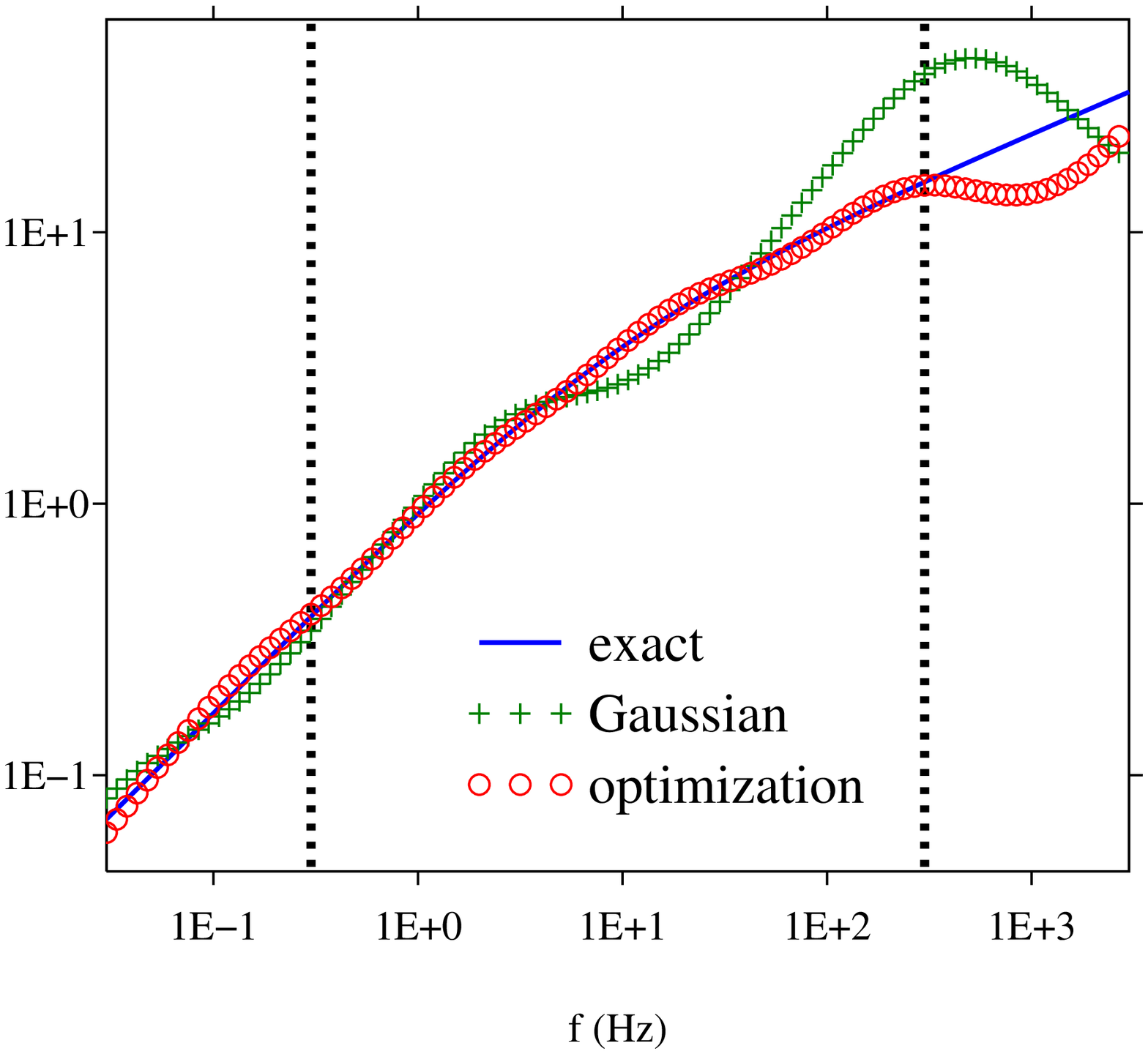}\label{QOpti4}}\\    
\subfloat[Phase velocity $c(f)$]{\includegraphics[width=0.5\textwidth]{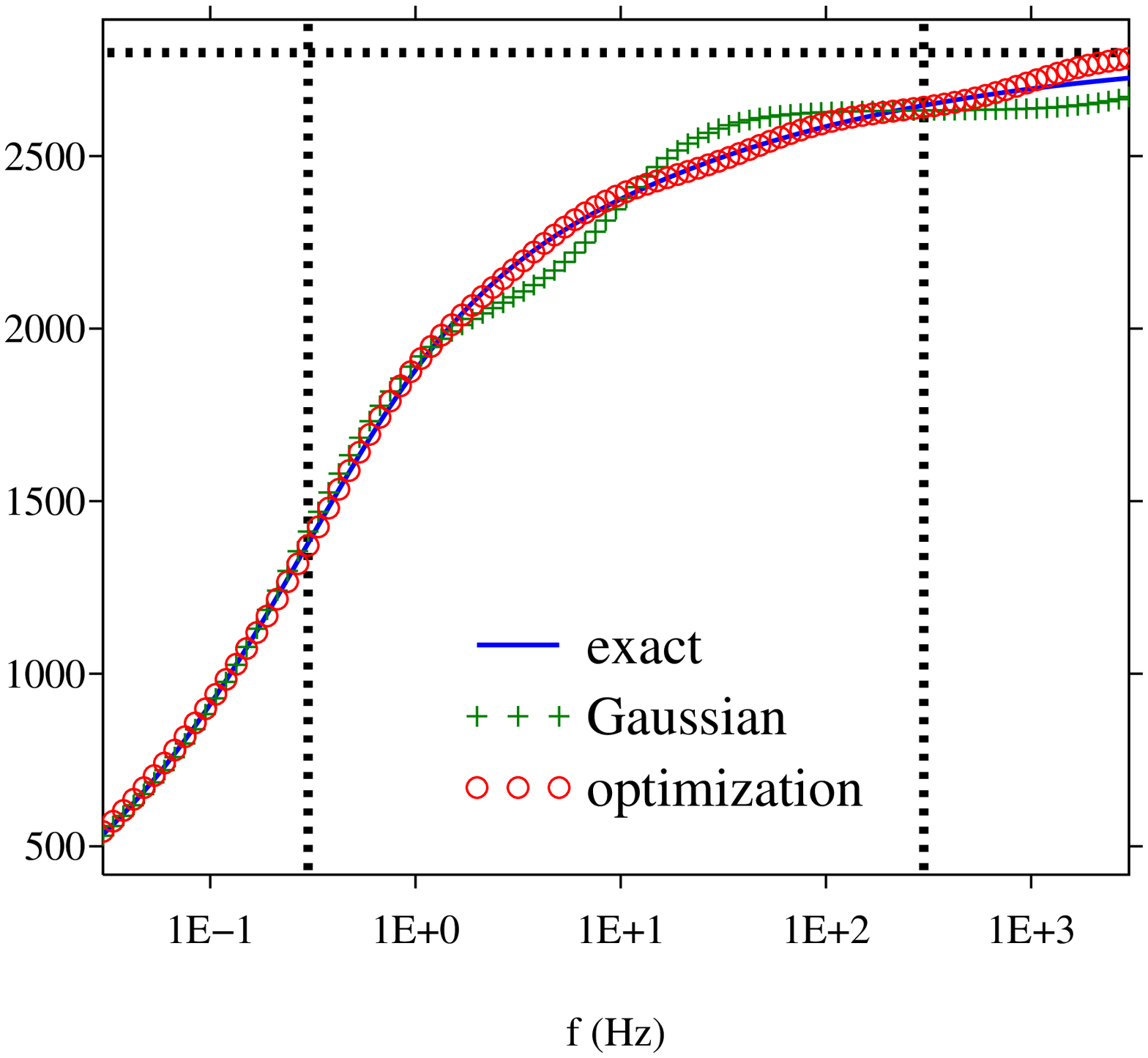}\label{CelOpti4}}
\subfloat[Attenuation $\zeta(f)$]{\includegraphics[width=0.5\textwidth]{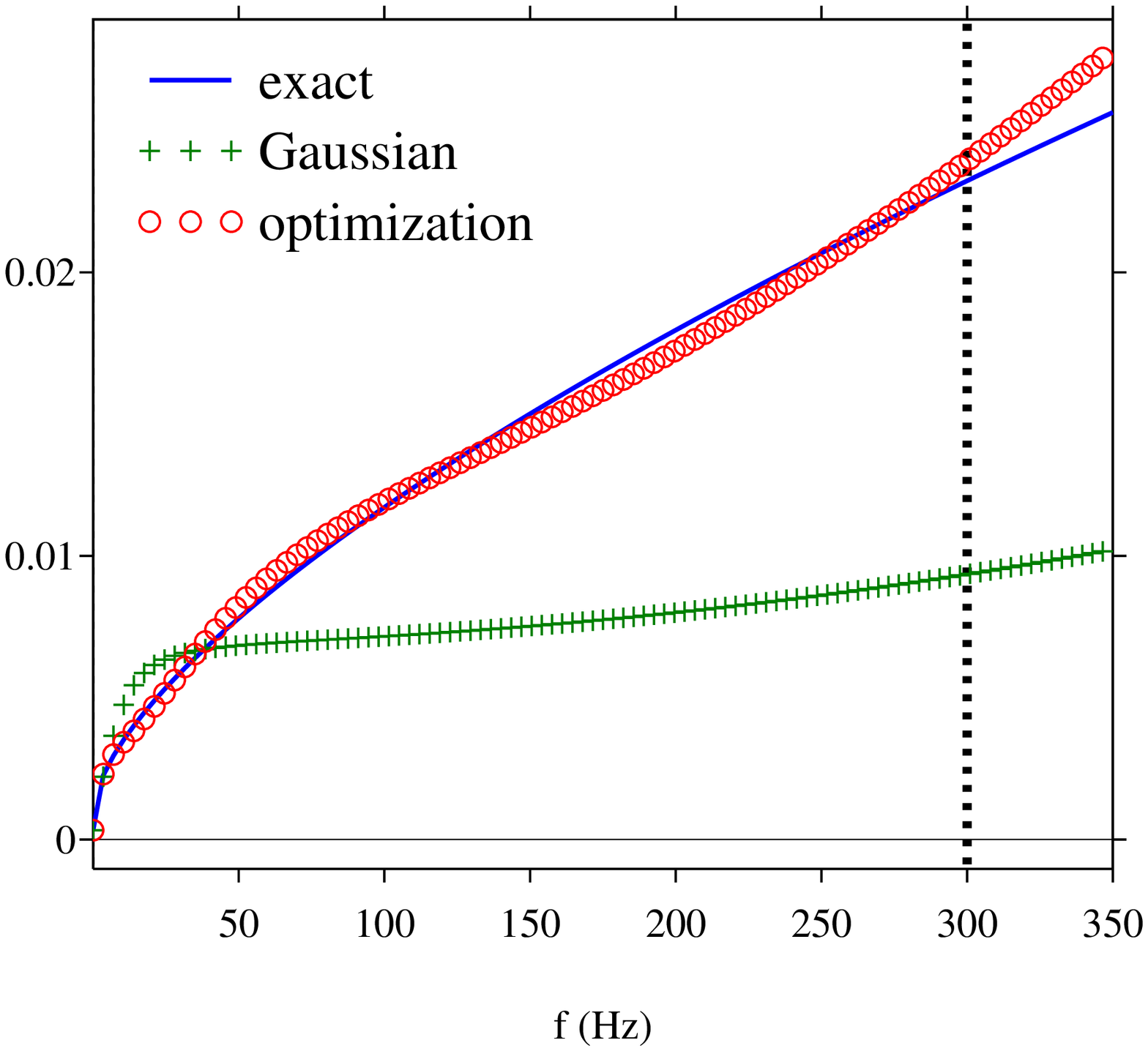}\label{AttOpti4}}
\caption{Comparison between exact observables of the Andrade model and their approximate counterparts, for $L=4$ memory variables. The physical parameters are given in Table \ref{TabParam}. Vertical dotted lines delimit the optimization frequency-band. The horizontal dotted line in panel (c) denotes the high-frequency limit $c_{\infty}$.}
\label{FigApprox}
\end{figure}

\begin{figure}[htp]		
\centering
\subfloat[$L=4$]{\includegraphics[width=0.5\textwidth]{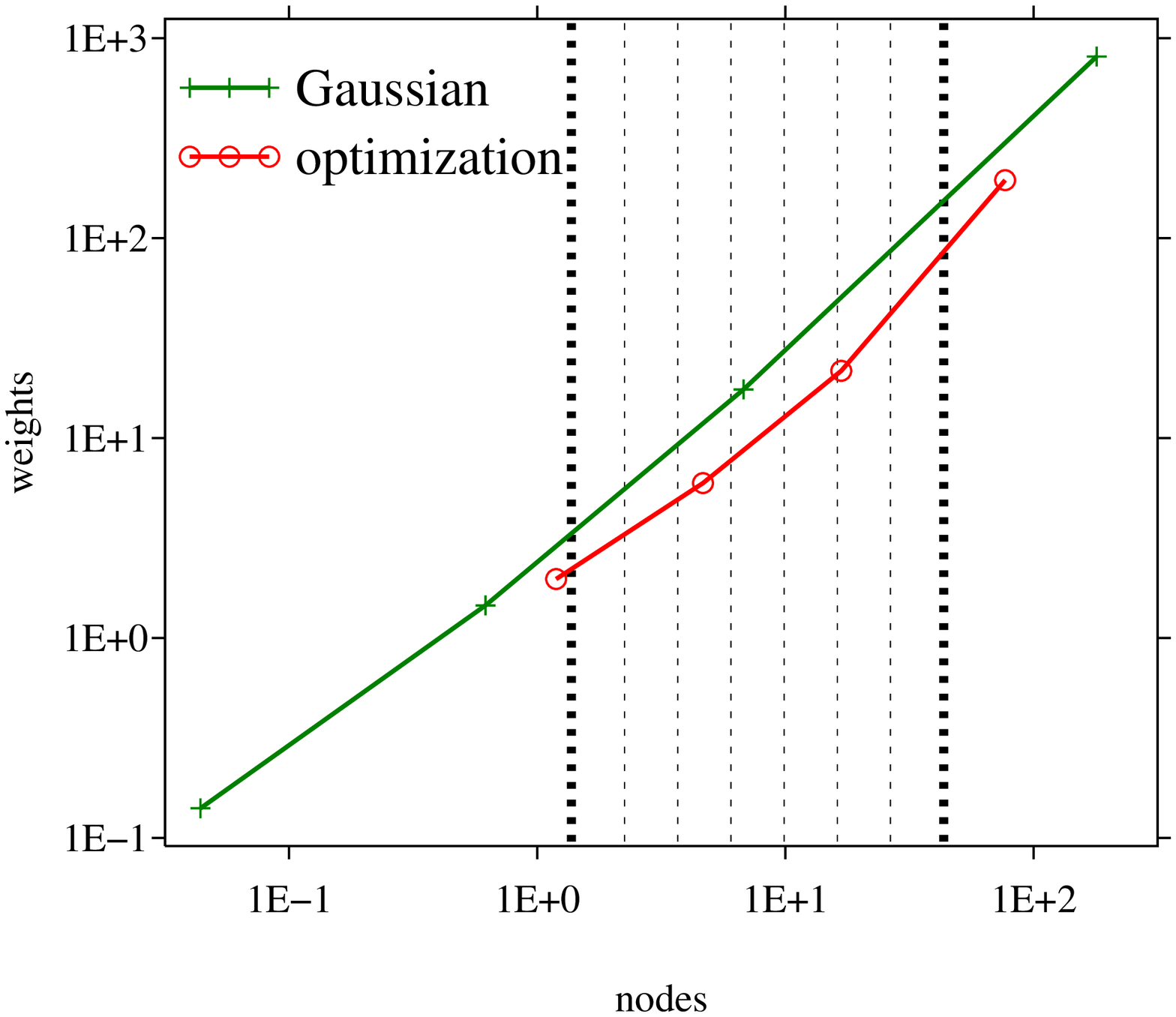}\label{FigNodes4}}
\subfloat[$L=8$]{\includegraphics[width=0.5\textwidth]{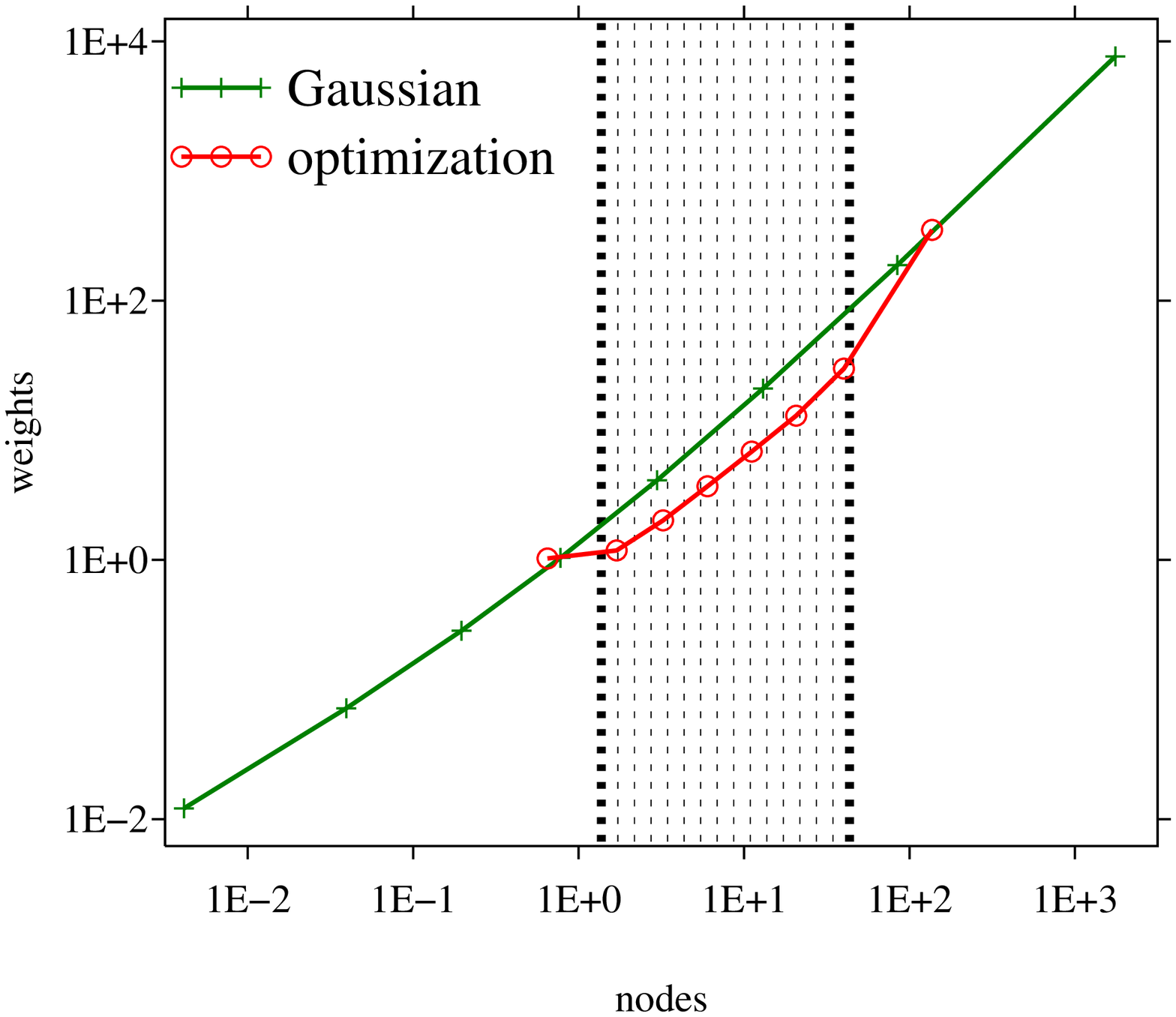}\label{FigNodes8}}
\caption{Quadrature coefficients ${\{(\mu_\ell,\theta_\ell)\}}_\ell$ for the two approaches considered. Vertical dotted lines denote the $K=2L$ scaled optimization angular frequencies $\sqrt{\omega_k}$.}
\label{FigNodes}
\end{figure}

On Figure \ref{FigNodes} are represented the $L=4$ and $L=8$ quadrature coefficients, i.e. nodes $\theta_\ell$ with corresponding weights $\mu_\ell$, for the two methods considered. Note that, according to \eqref{init:val}, the values provided by the Gaussian approach are used as initial guesses in the minimization \eqref{OptiNL}. The scaled optimization angular frequencies $\sqrt{\omega_k}$ for $k=1,\dots,K$ are also shown for the purposes of comparison. Remarkably, the computed optimal nodes do not coincide with equidistributed nodes along the optimization frequency-band, a repartition which is prescribed in the commonly employed approach of \cite{Emmerich}.

\begin{figure}[hb!]		
\centering
\subfloat[$L=4$]{\includegraphics[width=0.5\textwidth]{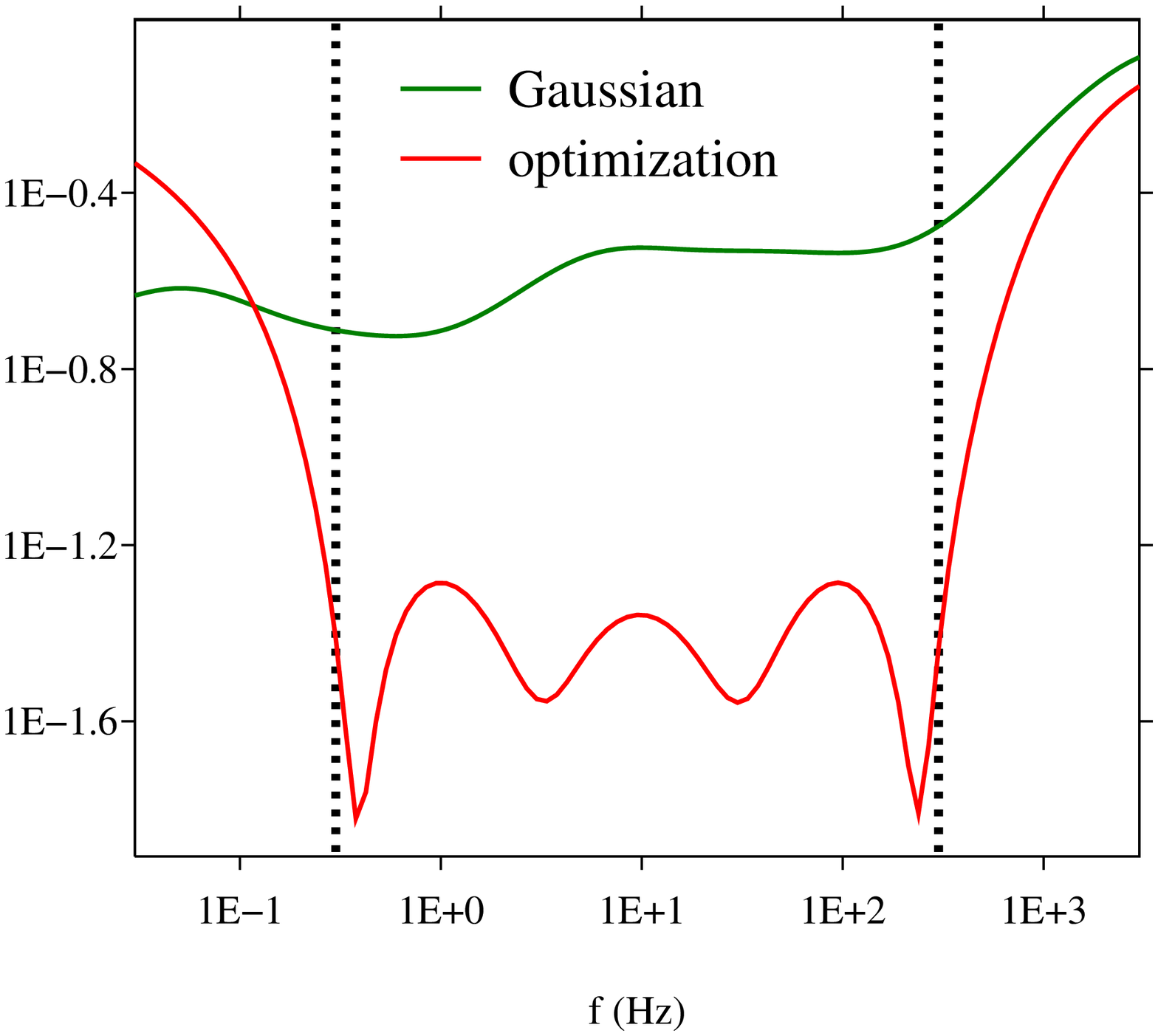}\label{FigErreur4}}
\subfloat[$L=8$]{\includegraphics[width=0.5\textwidth]{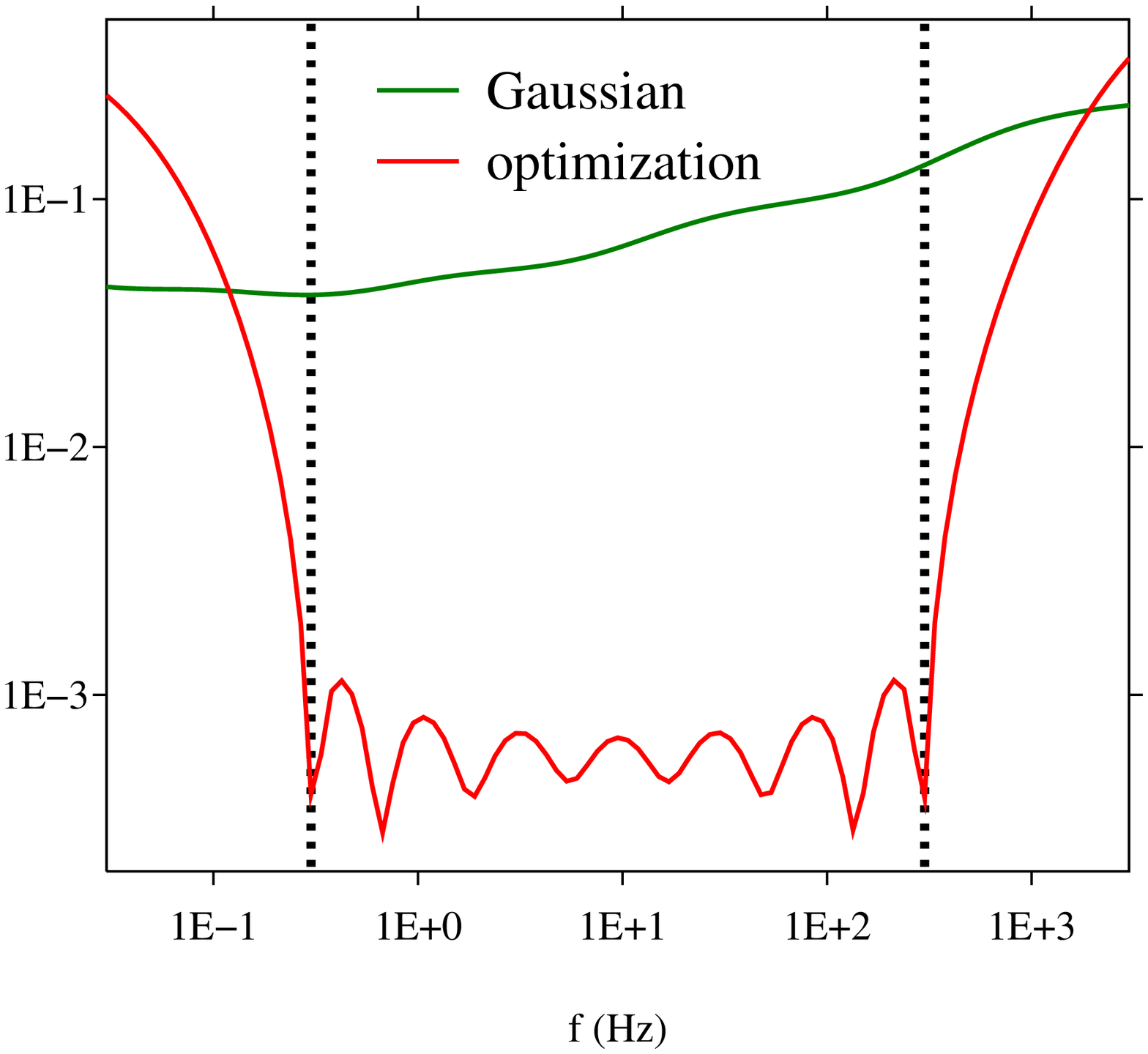}\label{FigErreur8}}
\caption{Computed error $|\frac{{\tilde \kappa}(f)}{\kappa(f)}-1|$. Vertical dotted lines delimit the optimization frequency-band.}
\label{FigErreur}
\end{figure}

The corresponding model error defined as $|\frac{{\tilde \kappa}(\omega)}{\kappa(\omega)}-1|$ and associated with the minimization problem \eqref{OptiNL} is displayed in Figure \ref{FigErreur}, for $L=4$ (Fig. \ref{FigErreur4}) and $L=8$ (Fig. \ref{FigErreur8}) diffusive variables. For a given quadrature method, the results are clearly improved as $L$ increases. For a given $L$, the optimization provides more accurate results compared to the Gaussian quadrature over the frequency band of interest which is delimited by vertical dotted lines.


\subsection{Validation of the numerical scheme}\label{SecExpWP}

While $\fs=0$ in \eqref{EDP2}, the source in \eqref{EDP1} is imposed at point $x_s$ as $\fu(x,t)=F(t)\,\delta(x-x_s)$ where $F(t)$ is the $C^6$ function
\begin{equation}
F(t) = 
\left\lbrace 
\begin{aligned}
& \sin\,(\omega_ct) - \frac{21}{32}\,\sin\,(2\,\omega_ct) + \frac{63}{768}\,\sin\,(4\,\omega_c t)-\frac{1}{512}\,\sin\,(8\,\omega_c t) \text{ if}\;0\leq t\leq \frac{1}{f_c},\\
& 0 \qquad \text{otherwise}
\end{aligned}
\right. 
\label{JKPS_C6}
\end{equation}
with the central frequency $f_c={\omega_c}/{2\,\pi}=30$ Hz. Moreover, the domain is discretized with $N_x=400$ nodes and the diffusive approximation is computed by constrained optimization with $L=4$ memory variables and thus $K=8$ optimization frequencies. The CFL condition is chosen so that $c_\infty \Delta t/\Delta x=0.95$ and the time integration is performed up to final time $t_f=200\Delta t\approx67$ ms based on the fourth order ADER scheme, see Sec. \ref{SecNumScheme}. Following Section \ref{SecPDEana}, the semi-analytical solution of the Andrade--DA model is computed by discrete inverse Fourier transform on $2048$ modes, with uniform frequency step $\Delta f=0.15$ Hz. The solution is recorded at each time step at receivers located at $x_r=220+40\,(r-1)$ for $r=1,\dots,5$.

\begin{figure}[th!]	
\centering
\subfloat[Snapshots at $t_f$]{\includegraphics[width=0.48\textwidth]{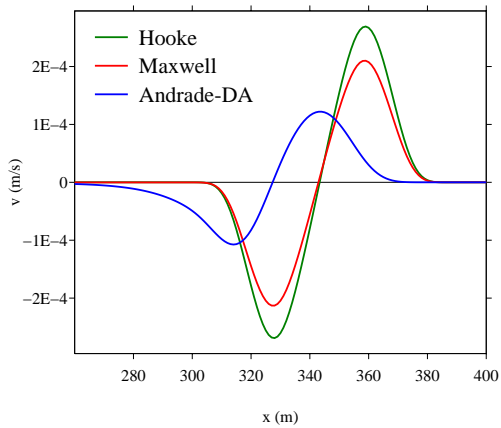}\label{HookMaxAnd}}\hspace{0.02\textwidth}
\subfloat[Time evolution for the Andrade--DA model]{\includegraphics[width=0.48\textwidth]{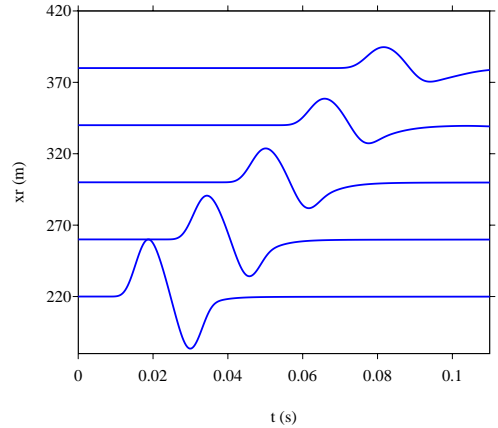} \label{SismoNew2}}
\caption{Time-domain numerical simulations of wave propagation.}
\label{FigSimus} 
\end{figure}

Figure \ref{FigSimus} displays snapshots of forward propagating waves from the source point $x_s=200$. The numerical solutions associated with various values of the attenuation parameters in \eqref{CreepDA} are plotted on Fig. \ref{HookMaxAnd}; namely Hooke model (i.e. purely elastic case which may be obtained in the limit $\eta=+\infty$ and setting $A=0$), Maxwell model ($A=0$, $\eta=10^{9}$), and Andrade--DA model ($A=2\,10^{-9}$, $\eta=10^{9}$). As predicted by the dispersion analysis of sections \ref{SecModelWave} and \ref{SecModelDA}, the phase velocity of the Andrade--DA model, as this of its original version, is lower than in the elastic case, which explains the observed delay. Figure \ref{SismoNew2} shows a seismogram corresponding to the Andrade--DA model in order to highlight attenuation and dispersion of the waveform.

\begin{figure}[bh!]
\centering
\subfloat[Semi-analytical and numerical solutions]{\includegraphics[width=0.48\textwidth]{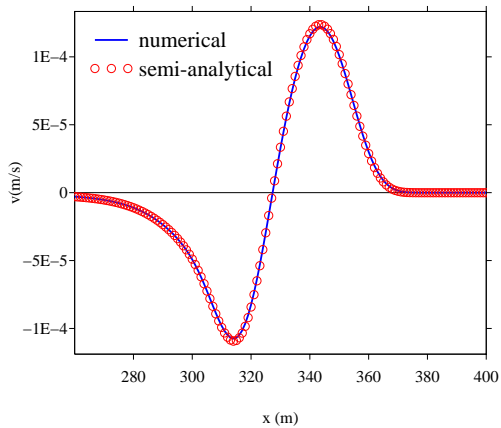}\label{AnalytiqueNew}}\hspace{0.02\textwidth}
\subfloat[Convergence measurements with numbers indicating slopes between two neighboring points]{\includegraphics[width=0.48\textwidth]{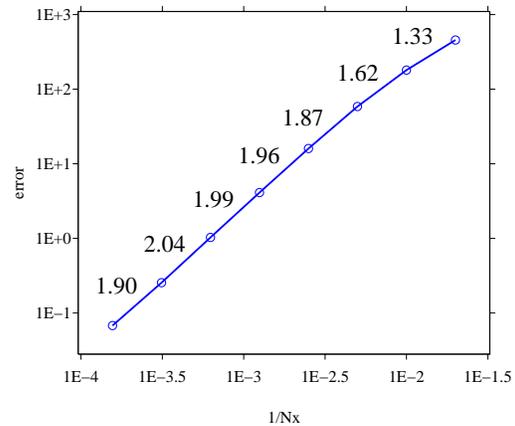}\label{FigAnalytique2} }
\caption{Validation of the numerical scheme for the Andrade--DA model.}
\label{FigAnalytique} 
\end{figure}

Figure \ref{FigAnalytique} compares the semi-analytical and numerical solutions of Andrade--DA model corresponding to equations \eqref{sol:v:phi} and \eqref{EDP}, respectively. Figure \ref{FigAnalytique2} presents convergence measurements done for various discretizations, varying the numbers of nodes in the interval $N_x=50$ to $6400$. Order $2$ is reached, confirming the theoretical results of Section \ref{SecNumScheme}.


\section{Conclusion}\label{SecConclu}

Wave propagation phenomena associated with a fractional viscoelastic medium are investigated in this study. The Andrade model is used as a prototypical reference constitutive equation as it satisfactorily describes the transient behaviors of metals and geological media. A diffusive representation of the featured non-local fractional derivative term is introduced to convert the associated convolution product into an integral of a function satisfying a local ordinary differential equation. Based on a quadrature approximation of this integrated term, a system of local partial differential equations is finally obtained and is shown to be well-suited for a numerical implementation.

The system at hand is investigated and it is demonstrated that its well-posedness requires the positiveness of the weights associated with the quadrature scheme. To compute the quadrature coefficients, two numerical methods are combined: a polynomial Gaussian approach to get an initial guess jointly with a constrained optimization to approximate the Andrade model compliance over a frequency-band of interest. It is shown that the properties of the original Andrade model are well approximated by those of the computed Andrade--DA model. Finally, an explicit time-domain finite-difference scheme is described and implemented. Corresponding wave propagation numerical experiments are presented and the efficiency of the proposed approach is highlighted. The main point of this article is that using a diffusive approximation of a fractional derivative term, entering a given viscoelastic constitutive equation, yields a sound mathematical model, that is also easily tractable numerically to perform wave propagation simulations. 

To focus on this message, a simple but realistic fractionally-damped viscoelastic model within a unidimensional and homogeneous configuration has been considered. Its dynamical behavior is described by a first-order hyperbolic system which extension to higher spatial dimensions or heterogeneous media is straightforward. Moreover, efficient numerical methods are currently available and can be directly employed to perform corresponding time-domain simulations. Alternatively, arbitrary-shaped material discontinuities within piecewise-homogeneous 2D Andrade media can be handled using an immersed interface method \cite{Chiavassa11}. Work is currently done on the subject. 

Another line of research concerns extension of the proposed approach to other fractional viscoelastic model, such as the fractional Kelvin-Voigt model \cite{Caputo67,Caputo11} or the fractional Zener model \cite{Pritz96,Nasholm:2013}. More sophisticated models could also be investigated, such as nonlinear fractional viscoelasticity \cite{Hanyga07} or nonlocal models in space \cite{Hanyga12a}. 

\paragraph{Acknowledgements.} The work of Abderrahmin Ben Jazia has been funded by the Ecole Centrale de Marseille, France, for which special thanks are addressed to Guillaume Chiavassa. The authors are thankful to Emilie Blanc for fruitful discussions


\bibliographystyle{elsarticle-num}
\bibliography{biblio}

\begin{thebibliography}{10}
\expandafter\ifx\csname url\endcsname\relax
  \def\url#1{\texttt{#1}}\fi
\expandafter\ifx\csname urlprefix\endcsname\relax\def\urlprefix{URL }\fi
\expandafter\ifx\csname href\endcsname\relax
  \def\href#1#2{#2} \def\path#1{#1}\fi

\bibitem{Andrade:1910iq}
E.~N. d.~C. Andrade, On the viscous flow in metals, and allied phenomena, Proc.
  Roy. Soc. A 84~(567) (1910) 1--12.

\bibitem{Szabo:94}
T.~L. Szabo, Time domain wave equations for lossy media obeying a frequency
  power law, J. Acoust. Soc. Am. 96 (1994) 491--500.

\bibitem{Szabo:95}
T.~L. Szabo, Causal theories and data for acoustic attenuation obeying a
  frequency power law, J. Acoust. Soc. Am. 97 (1995) 14--24.

\bibitem{Flanagan}
M.~P. Flanagan, D.~Wiens, {Attenuation of broadband P and S waves in Tonga:
  Observations of frequency dependent Q}, P. App. Geophys. 153 (1998) 345--375.

\bibitem{Lekic}
V.~Leki{\'c}, J.~Matas, M.~Panning, B.~Romanowicz, Measurement and implications
  of frequency dependence of attenuation, Earth And Planetary Science Letters
  282 (2009) 285--293.

\bibitem{Torvik}
P.~J. Torvik, R.~L. Bagley, On the appearance of the fractional derivative in
  the behavior of real materials, J. Appl. Mech. 51 (1985) 294--298.

\bibitem{Carcione}
J.~M. Carcione, {Wave Fields in Real Media: Wave propagation in Anisotropic,
  Anelastic and Porous Media}, Pergamon, 2001.

\bibitem{Podlubny}
I.~Podlubny, Fractional Differential Equations, Academic Press, 1999.

\bibitem{Mainardi}
F.~Mainardi, {Fractional Calculus and Waves in Linear Viscoelasticity. An
  Introduction to Mathematical Models}, Imperial College Press, 2010.

\bibitem{Holm}
S.~Holm, S.~P. N\"asholm, A causal and fractional all-frequency wave equation
  for lossy media, J. Acous. Soc. Am. 130~(4) (2011) 2195--2202.

\bibitem{Waters:2005}
K.~R. Waters, J.~Mobley, J.~G. Miller, {Causality-imposed (Kramers-Kronig)
  relationships between attenuation and dispersion}, IEEE Trans. Ultrason.
  Ferroelectr. Freq. Control. 52 (2005) 822--33.

\bibitem{Waters}
K.~R. Waters, M.~S. Hughes, J.~Mobley, G.~H. Brandenburger, J.~G. Miller, {On
  the applicability of Kramers-Kr\"onig relations for ultrasonic attenuation
  obeying a frequency power law}, J. Acoust. Soc. Am. 108 (2000) 556--563.

\bibitem{Emmerich}
H.~Emmerich, M.~Korn, Incorporation of attenuation into time-domain
  computations of seismic wave fields, Geophysics 52 (1987) 1252--1264.

\bibitem{Nasholm}
S.~P. N\"asholm, S.~Holm, Linking multiple relaxation, power-law attenuation,
  and fractional wave equations, J. Acous. Soc. Am. 130~(5) (2011) 3038--3045.

\bibitem{Andrade:1962uh}
E.~N.~d. Andrade, {On the validity of $t^{1/3}$ law of flow of metals}, Philos.
  Mag. 7~(84) (1962) 2003--2014.

\bibitem{Gribb:1998tf}
T.~T. Gribb, R.~F. Cooper, {Low-frequency shear attenuation in polycrystalline
  olivine: Grain boundary diffusion and the physical significance of the
  Andrade model for viscoelastic rheology}, J. Geophys. Res. Solid Earth
  103~(B11) (1998) 27267--27279.

\bibitem{Jackson:2010iv}
I.~Jackson, U.~H. Faul, {Grainsize-sensitive viscoelastic relaxation in
  olivine: Towards a robust laboratory-based model for seismological
  application}, Physics of the Earth and Planetary Interiors 183~(1-2) (2010)
  151--163.

\bibitem{Sundberg:2010p4848}
M.~Sundberg, R.~F. Cooper, A composite viscoelastic model for incorporating
  grain boundary sliding and transient diffusion creep: {C}orrelating creep and
  attenuation responses for materials with a fine grain size, Philos. Mag.
  90~(20) (2010) 2817--2840.

\bibitem{Bellis13a}
C.~Bellis, B.~Holtzman, Sensitivity of seismic measurements to
  frequency-dependent attenuation and upper mantle structure: an initial
  approach, submitted (2013).

\bibitem{Galucio}
A.~C. Galucio, J.-F. De\"{u}, R.~Ohayon, Finite element formulation of
  viscoelastic sandwich beams using fractional derivative operators, Comput.
  Mech. 33~(4) (2004) 282--291.

\bibitem{Matignon:proc}
D.~Matignon, Stability properties for generalized fractional differential
  systems, ESAIM: Proc. 5 (1998) 145--158.

\bibitem{Yuan02}
L.~Yuan, O.~P. Agrawal, A numerical scheme for dynamic systems containing
  fractional derivatives, J. Vib. Acoust. 124 (2002) 321--324.

\bibitem{Diethelm08}
K.~Diethelm, An investigation of some nonclassical methods for the numerical
  approximation of {C}aputo-type fractional derivatives, Numer. Algor. 47
  (2008) 361--390.

\bibitem{Birk10}
C.~Birk, C.~Song, An improved non-classical method for the solution of
  fractional differential equations, Comput. Mech. 46 (2010) 721--734.

\bibitem{Haddar}
H.~Haddar, J.-R. Li, D.~Matignon, Efficient solution of a wave equation with
  fractional-order dissipative terms, J. Comput. Appl. Math. 234 (2010)
  2003--2010.

\bibitem{Deu}
J.-F. De\"{u}, D.~Matignon, {Simulation of fractionally damped mechanical
  systems by means of a Newmark-diffusive scheme}, Comput. Math. Appl. 59
  (2010) 1745--1753.

\bibitem{Wismer}
M.~G. Wismer, Finite element analysis of broadband acoustic pulses through
  inhomogenous media with power law attenuation, J. Acous. Soc. Am. 120 (2006)
  3493--3502.

\bibitem{Caputo11}
M.~Caputo, J.~M. Carcione, F.~Cavallini, {Wave simulation in biologic media
  based on the Kelvin-Voigt fractional-derivative stress-strain relation},
  Ultrasound Med. Biol. 37~(6) (2011) 996--1004.

\bibitem{LeVeque92}
R.~J. LeVeque, {Numerical Methods for Conservation Laws}, 2nd Edition,
  Birkha\"user-Verlag, 1992.

\bibitem{Hanyga13a}
A.~Hanyga, On wave propagation in viscoelastic media with concave creep
  compliance, arXiv:1302.1797v2 (2013).

\bibitem{Hanyga13b}
A.~Hanyga, Wave propagation in linear viscoelastic media with completely
  monotonic relaxation moduli, arXiv:1302.0402v2 (2013).

\bibitem{NRPAS}
B.~Flannery, W.~Press, S.~Teukolsky, W.~Vetterling, Numerical Recipes in C: the
  Art of Scientific Computing, Cambridge University Press, 1992.

\bibitem{TheseBlanc}
E.~Blanc, {Numerical modeling of transient poroelastic waves: the Biot-JKD
  model with fractional derivatives}, Ph.D. thesis, Aix-Marseille Universit\'e
  (2013).

\bibitem{Blanc13}
E.~Blanc, G.~Chiavassa, B.~Lombard, {Biot-JKD model: simulation of 1D transient
  poroelastic waves with fractional derivatives}, J. Comput. Phys. 237 (2013)
  1--20.

\bibitem{Kappel00}
F.~Kappel, A.~Kuntsevich, {An implementation of Shor's r-algorithm}, Comput.
  Optim. Appl. 15~(2) (2000) 193--205.

\bibitem{Rekik11}
A.~Rekik, R.~Brenner, Optimization of the collocation inversion method for the
  linear viscoelastic homogenization, Mech. Res. Com. 38 (2011) 305--308.

\bibitem{Shor85}
N.~Shor, {Minimization Methods for Non-Differentiable Functions},
  Springer-Verlag, 1985.

\bibitem{Moler}
C.~Moler, C.~Van~Loan, Nineteen dubious ways to compute the exponential of a
  matrix, twenty-five years later, SIAM Review 45 (2003) 3--49.

\bibitem{Schwartzkopff04}
T.~Schwartzkopff, M.~Dumbser, C.~Munz, Fast high-order ader schemes for linear
  hyperbolic equations, J. Comput. Phys. 197~(2) (2004) 532--539.

\bibitem{Chiavassa11}
G.~Chiavassa, B.~Lombard, {Time domain numerical modeling of wave propagation
  in 2D heterogeneous porous media}, J. Comput. Phys. 230 (2011) 5288--5309.

\bibitem{Caputo67}
M.~Caputo, {Linear models of dissipation whose Q is almost frequency
  independent, part II}, Geophys. J. R. Astr. Soc. 13 (1967) 529--–539.

\bibitem{Pritz96}
T.~Pritz, Analysis of four-parameter fractional derivative model of real solid
  materials, J. Sound Vib. 100 (1996) 1301--1315.

\bibitem{Nasholm:2013}
S.~P. N\"asholm, S.~Holm, {On a fractional Zener elastic wave equation}, Fract.
  Calc. Appl. Anal. 16 (2013) 26--50.

\bibitem{Hanyga07}
A.~Hanyga, Fractional-order relaxation laws in non-linear viscoelasticity,
  Continuum Mech. Thermodyn. 19 (2007) 25--36.

\bibitem{Hanyga12a}
A.~Hanyga, M.~Seredynska, Spatially fractional-order viscoelasticity,
  non-locality, and a new kind of anisotropy, J. Math. Phys. 53 (2012) 052902.

\end{thebibliography}

\end{document}